\documentclass[12pt,a4paper]{amsart}
\usepackage[utf8]{inputenc}
 \usepackage[english]{babel}
\usepackage{stmaryrd}
\usepackage[a4paper,inner=2.8cm,outer=2.5cm,top=2.8cm,bottom=2.5cm]{geometry}               
\usepackage{amsmath,amssymb,amscd,amsthm,amsfonts,anyfontsize,color,enumerate,fix-cm,layout,lipsum,lpic,mathrsfs,mdwlist,stmaryrd,tensor,tikz,thmtools,xspace,comment}
\usepackage{float}
\usepackage{pdfpages}
\usepackage{graphicx}
\usepackage{hyperref}
\usepackage[all]{xy}
\usepackage{mathtools}
\usepackage[bottom]{footmisc}
\usepackage{ifthen}
\usepackage{xfrac}
\usepackage{tikz-cd}
\theoremstyle{plain}                 
\newtheorem{theorem}{Theorem}[section]     
\newtheorem{proposition}[theorem]{Proposition}

\newtheorem{problem}[theorem]{Problem}
\newtheorem{lemma}[theorem]{Lemma}    
\theoremstyle{definition}           
\newtheorem{definition}[theorem]{Definition}   
\newtheorem{example}[theorem]{Example}
\newtheorem{algorithm}[theorem]{Algorithm}
\theoremstyle{remark}       
\newtheorem{remark}[theorem]{Remark} 
\newtheorem{conjecture}[theorem]{Conjecture} 
\hypersetup{colorlinks=true,linkcolor=blue,citecolor=purple,filecolor=magenta,urlcolor=cyan}

\begin{document}


\author[A. A. Kazakov]{A. A.~Kazakov}
\address{A.~K.:  
Centre of Integrable Systems, P. G. Demidov Yaroslavl State University, Sovetskaya 14, 150003, Yaroslavl, Russia.}
\email{anton.kazakov.4@mail.ru}
\title{Inverse problems related to  electrical networks and the geometry of non-negative Grassmannians}

\begin{abstract}
We provide a new solution to the classical black box problem (the discrete Calderón problem) in the theory of circular electrical networks. Our approach  is based on   the explicit embedding of electrical networks into  non-negative Grassmannians and   generalized chamber ansatz for it. Also, we reveal the  relation of this problem with the combinatorial properties  of spanning groves and the theory of   totally non-negative matrices. 
\end{abstract}
\maketitle

\tableofcontents

\hspace{0.1in}

\section{Introduction} \label{chap1}
\subsection{Organization of the paper}
An electrical network $\mathcal{E}(G, w)$  consists of a planar  graph $G$ equipped with positive  edge weights $w$ denoting their conductivities. The nodes of $\mathcal{E}(G, w)$ are partitioned into two sets: inner nodes and boundary nodes.  By applying voltages     $\textbf{U}:V_{B}\to \mathbb{R}$ at the  boundary nodes and using  Ohm's and Kirchhoff's laws, the voltages at all nodes $U: V \to \mathbb{R}$ are uniquely determined.  These voltage extensions  give rise to two natural  network characteristics: the electrical response matrix and the effective resistance matrix.

Combining the generalized Temperley trick with the relation between the entries of these matrices and spanning forests enables the construction of an embedding into the Grassmannian $Gr_{\geq 0}(n-1, 2n)$, see Sections \ref{thelamembsec} and \ref{efsec}.    
This embedding  can be  explicitly defined via response matrices or effective resistance matrices, see Theorem \ref{th: main_gr} and Theorem \ref{th: main_gr}.
That allows us to apply  the Postnikov theory (see Section \ref{postteorsec})  of  non-negative Grassmannians to the study of electrical networks.   Particularly, drawing upon the Postnikov theory,  we  propose (see Section \ref{secan}) a novel algorithmic   solution to the black box problem, which is the main result of our paper (\textit{note that a similar approach has been independently  studied in \cite{Ter}}):
\begin{proposition}
    Consider a minimal electrical network $\mathcal{E}(G, w)$ and its associated point $\mathcal{L}(\mathcal{E})$ of non-negative Grassmannian $Gr_{\geq 0}(n-1, 2n)$. Then,  each edge conductivity  can be  expressed as monomial ratios of  a special set of Plücker coordinates of the twisted point  $\tau\bigl(\mathcal{L}(\mathcal{E})\bigr)$ of $Gr_{\geq 0}(n-1, 2n)$ corresponding to weights of minimal spanning groves. The indices of these  Plücker coordinates are determined by applying the  Scott rule to the minimal Lam model, which is constructed from $\mathcal{E}(G, w)$ using the generalized Temperley trick.
\end{proposition}
Unlike  other existing methods, our solution  is applicable for any network and has a non-recursive nature, which we anticipate   might lead to  more numerically stable approaches to    the Calderón problem (see  Sections \ref{bl-box-sec}) and some problems in phylogenetic network theory, see references in Section \ref{efsec}.  

Finally, in Section \ref{luzsec}  we reveal the   principal relation of the black box problem  to the classical theory of   totally non-negative  matrices.

\textbf{Acknowledgments.} 
The work on Sections \ref{chap1}, \ref{sec:blbox}, \ref{secan} was funded  by the Russian Science Foundation project No. 20-71-10110 (https://rscf.ru/en/project/23-71-50012) which finances the work of the author at P. G. Demidov Yaroslavl State University. The work on Sections \ref{luzsec} was supported by the  Ministry of Science and Higher Education of the Russian Federation (Agreement No. 075-02-2025-1636).

The author is grateful to  Vassily Gorbounov, Dmitry Talalaev, Sotiris Konstantinou-Rizos, and Lazar Guterman for their invaluable discussions and careful reading of this work. Special thanks are due to Lazar Guterman for his assistance with the figures.

\subsection{Electrical network theory} 
\label{sec:circ}
We start with a brief introduction into the   theory of  electrical networks  following \cite{CIM}, \cite{CIW} and \cite{K}.
\begin{definition} 
An \textbf{electrical network} $\mathcal{E}(G,  w)$ is essentially a planar graph $G(V, E)$,  embedded into a disk and  equipped with a  conductivity function $w: E(G) \to \mathbb{R}_{> 0}$, which satisfies the following conditions:  
\begin{itemize}
    \item All nodes are divided into the set of inner nodes $V_I$ and the set of boundary nodes $V_B$ and  each boundary node lies on the boundary circle;
    \item Boundary nodes are enumerated clockwise from $1$ to $|V_B|:=n.$ Inner nodes are enumerated arbitrarily from $n+1$ to $|V|$;
    \item An edge weight  $w(v_iv_j)=w_{ij}$  denotes the conductivity of this edge. 
\end{itemize}

\end{definition}

Consider an  electrical network $\mathcal{E}(G,  w)$ and apply voltages   $\textbf{U}:V_{B}\to \mathbb{R}$ to its boundary nodes $V_B$. Then these boundary voltages induce the unique harmonic extension on all vertices $U: V \to \mathbb{R},$ which  might be found out by Ohm's and Kirchhoff's laws:
 \begin{equation*}
     \sum \limits_{j \in V}w_{ij}\bigr(U(i)-U(j)\bigl)=0, \ \forall i \in V_I.
 \end{equation*}

 One of the main objects associated with  each harmonic extension  of  boundary voltages   $\textbf{U}$ are 
 boundary currents $\textbf{I}=\{I_1, \dots, I_n\}$ running through boundary nodes:
 \begin{equation*}
    I_k:=\sum \limits_{j \in V}w_{ij}\bigr(U(k)-U(j)\bigl), \ k\in \{1, \dots, n\}.
 \end{equation*}

The first cornerstone result in the theory of electrical networks is that boundary voltages and currents are related to each other linearly.  
\begin{theorem} \textup{\cite{CIW}} \label{aboutresp}
Consider an electrical network $\mathcal{E}(G, w).$ Then, there is a  matrix $M_R(\mathcal{E})=(x_{ij}) \in Mat_{n \times n}(\mathbb{R})$ such that the following holds:
$$M_R(\mathcal{E})\textbf{U}=\textbf{I}.$$
This matrix is called the \textbf{response matrix} of a network $\mathcal{E}(G, \omega)$ and it satisfies the following conditions:
\begin{itemize}
    \item Matrix $M_R(\mathcal{E})$ is symmetric;
    \item All non-diagonal entries $x_{ij}$ of $M_R(\mathcal{E})$ are non-positive;
    \item For each row $($column$)$ the sum of all its entries is equal to $0.$
    \item  Matrix $M_R(\mathcal{E})$ is  circular totally non-negative i.e.  for each $k \in \{1, \dots, n \}$ all $k \times k$ circular minors of $M_R(\mathcal{E})$ $($see  \cite{CIW} and \cite{K} for more details$)$ are signed non-negative: $(-1)^k\det M_R(\mathcal{E})(P; Q) \geq 0,$ for each circular pair $(P;Q).$
\end{itemize}
\end{theorem} 


The following five local network transformations given below are called the \textbf{electrical transformations}. 
\begin{figure}[H]
    \centering
    \includegraphics[width=1.0\textwidth]{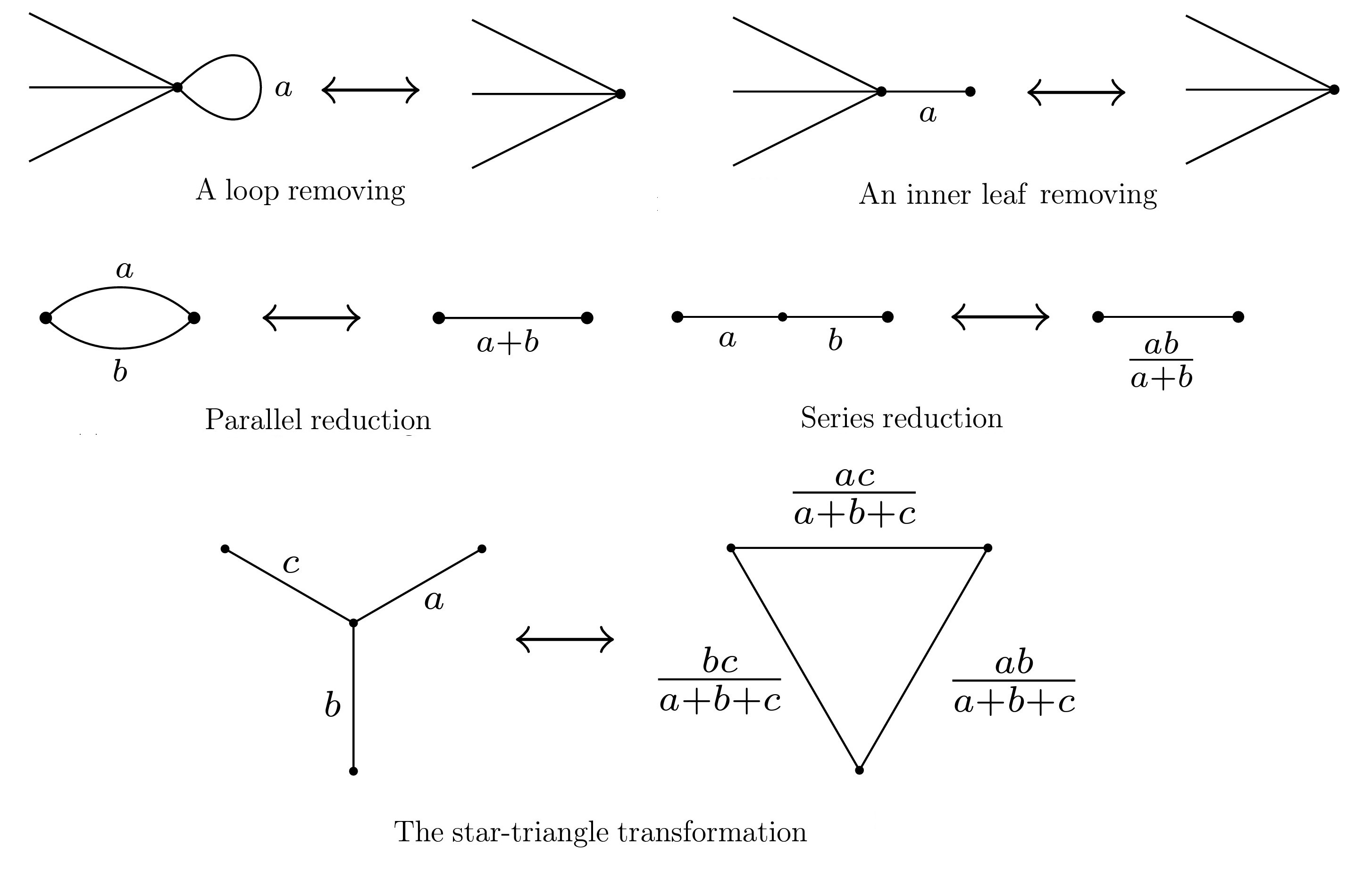}
 \hspace{-1.5cm}   \caption{Electrical transformations }
    \label{fig:el_trans}
\end{figure}

Two electrical networks are defined to be equivalent if they can be obtained from each other by a sequence of electrical transformations. The importance of these transformations is revealed by the following result.
\begin{theorem} \textup{\cite{CIM}} \label{gen_el_1}
The electrical transformations preserve the response matrix of an electrical network. Any two electrical networks that share the same response matrix are equivalent.
\end{theorem}

According to Theorem \ref{gen_el_1}, we will consider electrical networks  up to electrical transformations. We will denote by $E_n$ the set of all electrical networks with $n$ boundary nodes up to electrical transformations. In fact, each electrical network can be transformed  into a minimal network.

\begin{definition}
    A \textbf{median graph} of an electrical  network $\mathcal{E}(G, w) $  is a graph $G_M$ whose internal vertices are the midpoints of the edges of $G$ and two internal vertices are connected by an edge if the edges of the original graph $G$ are adjacent. The boundary vertices of $G_M$ are defined as the intersection of the natural extensions of the edges of $G_M$  with the boundary circle. Since the interior vertices of the median graph have degree four, we can define the \textbf{strands} of the median graph as  the  paths which always go straight through any degree four vertex. 
\end{definition}
\begin{figure}[H]
    \centering
    \includegraphics[width=0.4\textwidth]{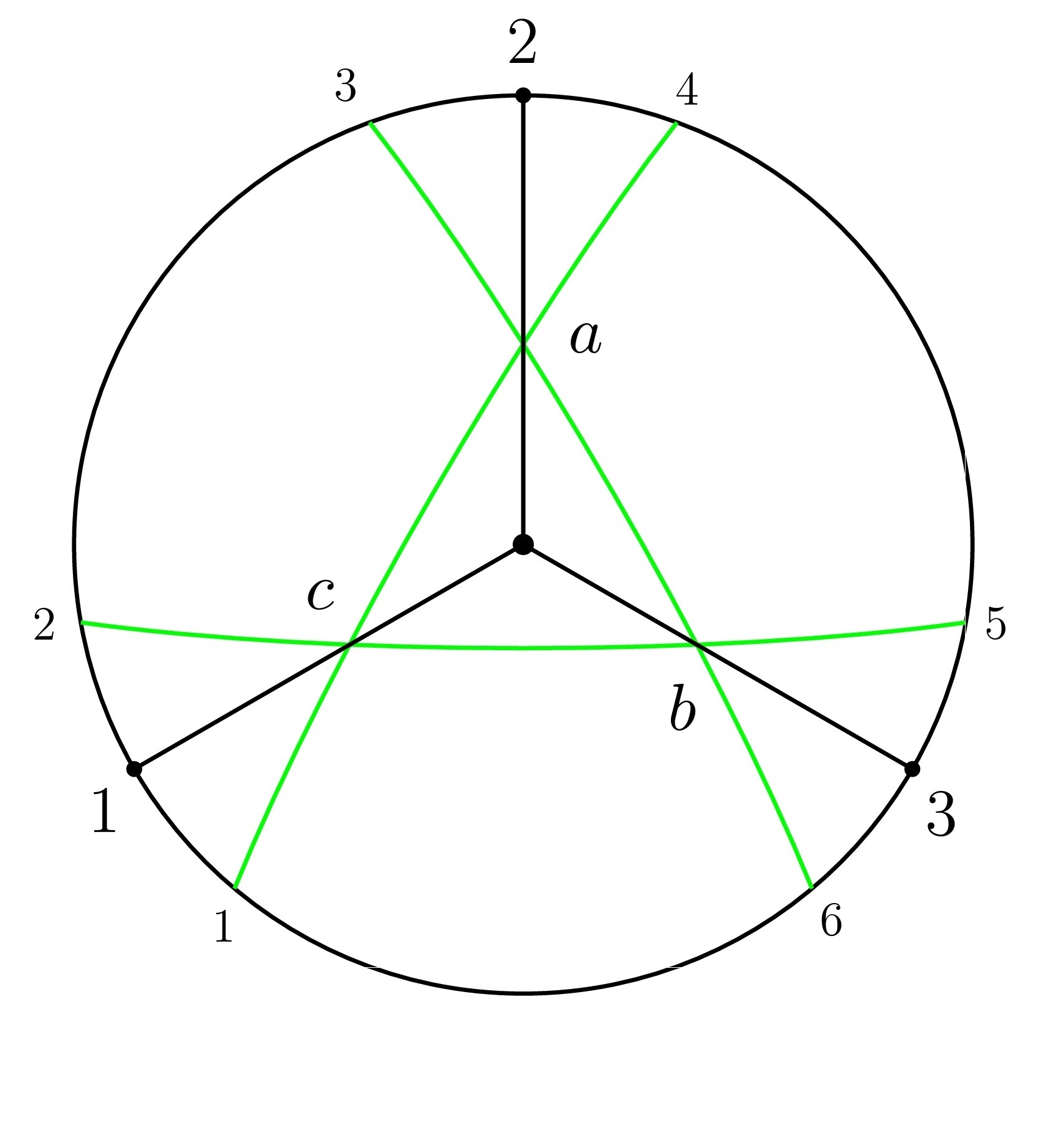}
    \caption{Star-shape network and its median graph}
    \label{fig:triangle}
\end{figure}
\begin{definition}
An  electrical network is called \textbf{minimal} if the strands of its median graph do not have self-intersections; each strand does not form a  closed loop; any two strands intersect at most once  i.e. the median graph has no  lenses, see Fig. \ref{fig:loop}.
\end{definition}

\begin{figure}[H]
    \centering
    \includegraphics[width=0.3\textwidth]{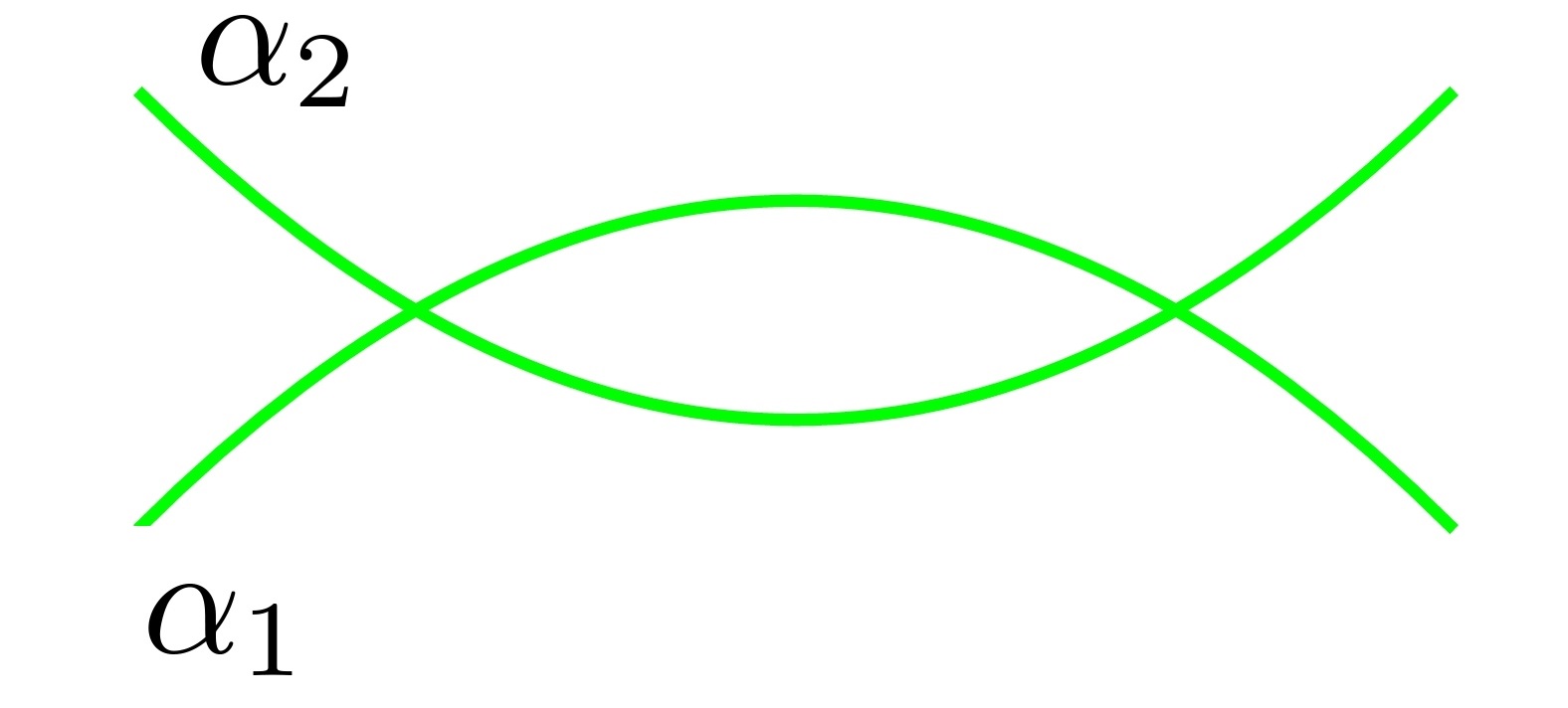}
    \caption{A lens obtained  by the intersection of strands  $
\alpha_1$ and  $\alpha_2$ }
    \label{fig:loop}
\end{figure}

\begin{theorem} \textup{\cite{CIW},\cite{K}}
\hskip 1pt
\begin{itemize}
    \item Each electrical network $\mathcal{E}(G, w) \in E_n$ is equivalent  to a minimal network;
    \item  Any two minimal  electrical networks which share  the same response matrix  can be transformed to each other only by the star-triangle transformation.
\end{itemize}
\end{theorem}

 \subsection{The black box problem for electrical networks} \label{bl-box-sec}

One of the most important and notoriously difficult problems in electrical network theory, which has significantly influenced the entire development of this field \cite{CIW}, is known as the \textbf{black box problem}.
\begin{problem}  \label{bl-box}    Consider an electrical network $\mathcal{E}(G, w)$ on a given graph $G.$ The conductivity function $w$ is required to be recovered by a known response matrix $M_R(\mathcal{E}).$
\end{problem}
The following theorem reveals the possibility of solving the black box problem:
\begin{theorem} \textup{\cite{CIM},\cite{CIW}}  \label{bl-box_th}	The conductivity function $w$ of a  minimal  electrical network $\mathcal{E}(G, w) \in E_n$ can be uniquely recovered by its response matrix $M_R(\mathcal{E}).$ 
\end{theorem}
Notably, all proofs of Theorem \ref{bl-box_th} are constructive and directly provide practical algorithms for solving Problem \ref{bl-box}. We will now highlight a few of the main ones:
\begin{itemize}
    \item The recursive  spike-bridge-pendant decomposition   \cite{CIW};   
    \item Solutions to Problem \ref{bl-box} for special (for instance, standard) electrical networks \cite{CIW}, \cite{KW space};
    \item Reducing Problem \ref{bl-box} to the black box problem in the Postnikov theory. This approach will be described in Section \ref{sec:blbox} (also see \cite{Ter}).
\end{itemize}
 Problem \ref{bl-box} plays a  central role  in electrical network theory because it enables the development of fast and stable numerical solutions to the well-known \textbf{Calderón problem} \cite{Uh1}, \cite{Uh2}. In fact, Problem \ref{bl-box} can be viewed as a discrete analog \cite{BDV} of a specific instance of the Calderón problem, known as the inverse problem of electrical impedance tomography in a circle. For completeness, we will briefly outline the formulation of this problem and its link to Problem \ref{bl-box}. 

Consider a conducting medium in a closed disk $D$ i.e. a positive bounded  function of conductivity $\sigma(x): D \to \mathbb{R}_{> 0}$ and an electric voltage $u(x): D \to \mathbb{R}$, which are related to each other by the following second-order elliptic partial differential equation:
\begin{equation*} \label{eq-diff}
    div(\sigma(x)\nabla u(x))=0, \ \forall x \in  int \ D,
\end{equation*}
with Dirichlet boundary data which is a known voltage  $V: \partial D \to \mathbb{R}$ applied to the boundary of $D$:
  \begin{equation*} \label{eq-bound}
    u(x)_{|\partial D}=V(x), \ \forall x \in \partial D.
\end{equation*}

The voltage $V(x)$ induces a Neumann boundary data, namely  currents $I_{V}(x)$  running   through     $\partial D$$: \partial D \to \mathbb{R}$ induced by $V(x)$ :
$$I_{V}(x)=(n, \sigma(x)\nabla u(x)_{|\partial D}),$$
where $n$ is the unit vector of the outer normal to $\partial D;$ and $(\cdot, \cdot)$ is the standard scalar product.

Suppose that for each boundary voltage $V(x)$ the corresponding induced  boundary currents $I_{V}(x)$ are known. Formally, it means that the operator $\Lambda^{DtN}$ from  Dirichlet data to  Neumann data is known:

\begin{equation*} 
   \Lambda^{DtN}V(x)=I_{V}(x).
\end{equation*}

\begin{problem}  \label{el-tom}
\textbf{Calderón problem:} It is required to recover a conductivity function $\sigma(x)$ using a known operator $\Lambda^{DtN}.$
\end{problem}

Omitting some important details,  we briefly describe the  idea of the  numerical approach to Problem \ref{el-tom} following \cite{BDV}. 

\begin{algorithm} \textup{\cite{BDV}}
\hskip 1pt
    \begin{itemize}
        \item Using the operator $\Lambda^{DtN},$ construct a matrix  $M$, which is a    response  of the general position (well-connected, see \cite{K}) minimal network $\mathcal{E}(G, w)$ on a given graph;
        \item Solve Problem \ref{bl-box} for  $\mathcal{E}$;
        \item Using the edge conductivities of $\mathcal{E}$  numerically found out a conductivity  $\sigma(x)$. 
    \end{itemize}
\end{algorithm}

\section{From electrical networks to non-negative Grassmanians} \label{sec:blbox}
\subsection{The Postnikov theory} \label{postteorsec}
The central idea of our paper is that the black box problem can be framed as a specific instance of the black box problem within the Postnikov theory, which we briefly describe following  \cite{L3} and \cite{Pos}.

\begin{definition}
The totally non-negative Grassmannian $\mathrm{Gr}_{\geq 0}(k, m)$ is a subset of the points in the Grassmannian $\mathrm{Gr}(k, m)$ whose Plücker coordinates $\Delta_I$ have the same sign or vanish.
\end{definition}
The key result of  the  Postnikov theory  is that each $X \in \mathrm{Gr}_{\geq 0}(k, m)$ can be parametrized by the Lam models.

\begin{definition}  \label{l1}
A planar  graph $\Gamma$ together with a weight function $\omega:E(\Gamma)\rightarrow \mathbb R_{> 0}$ is called a \textbf{Lam model} $N(\Gamma, \omega)$ if the following holds:
\begin{itemize}
\item $\Gamma$ is embedded into a disk $D$, and its nodes are divided on two sets, namely a set of  inner nodes $V_I$ and a set of boundary nodes $V_B;$ 
    \item The boundary nodes $V_B$ lay on  a boundary circle $\partial D$ and are numbered clockwise from $1$ to $n:=|V_B|$;
    \item The degrees of the boundary nodes are all equal to one;
    \item All the vertices of $\Gamma$ are colored by either black or white color and each edge is incident to the vertices of different colors.
\end{itemize}
Note that when illustrating Lam models, we will omit the colors of boundary nodes and suppress the weights of edges if they are equal to $1.$
\end{definition}

\begin{definition} A \textbf{dimer} $\Pi$ is a collection of edges of $\Gamma$ such that
\begin{itemize}
\item Each interior vertex is incident to one of the edges in $\Pi$;
\item Boundary vertices may or may not be incident to the edges in $\Pi$.
\end{itemize}
The weight $\mathrm{wt}(\Pi)$ of a dimer $\Pi$ is a product of  all  weights of its edges.  
\end{definition}
\begin{figure}[H]
     \hspace*{-10mm}
    \includegraphics[scale=0.1]{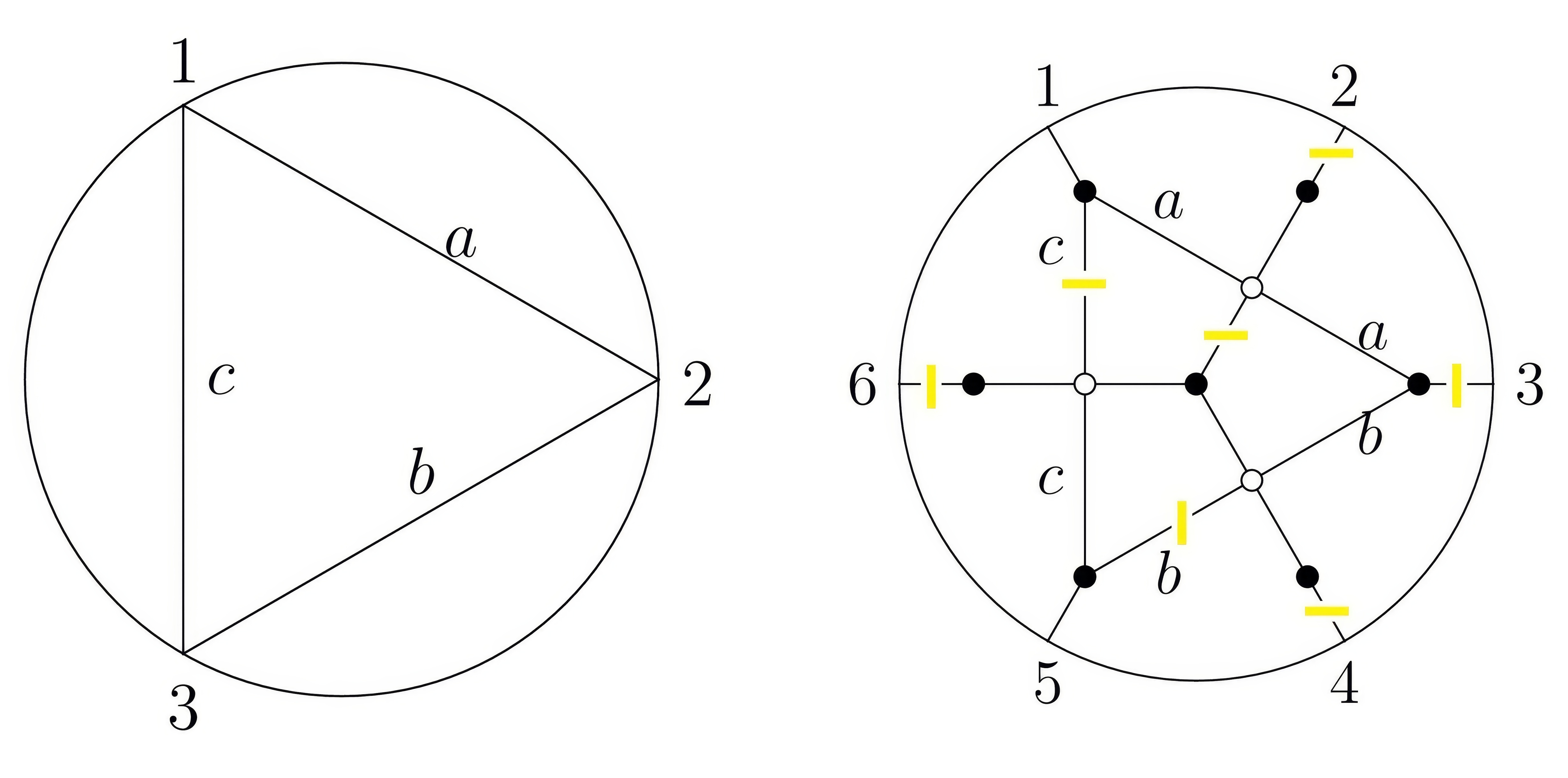}
     \caption{A Lam model with            $k(\Gamma)=2$ and its dimer $\Pi \in \Pi(15)$ with the weight $\mathrm{wt}(\Pi)$ equals to $ab$ }
   \end{figure}

Consider a Lam model  $N(\Gamma, \omega)$ and its arbitrary dimer $\Pi$.  It is not difficult to verify that the sum of a number of black boundary nodes  are {\it covered} by $\Pi$ with a  number of white boundary nodes  are {\it not covered} by $\Pi$ is independent of the choice of dimer $\Pi$. This sum can be calculated as follows: 

$$k(\Gamma)=\frac{1}{2}\left(n+\sum\limits_{v \in black}(\mathrm{deg}(v)-2)+\sum\limits_{v \in white}(2-\mathrm{deg}(v))\right),$$
where   summations are over black or white inner nodes of  $\Gamma.$ This observation leads to the following definition.
\begin{definition} \label{ld}
For each $I\subset V_I$ such that $|I|=k(\Gamma) $ the \textbf{boundary measurement} is given by the formula:
\begin{equation*}
\Delta_{I}^d=\sum \limits_{\Pi\in \Pi(I)}\mathrm{wt}(\Pi),
\end{equation*}
where  the summation goes over all dimers $\Pi \in  \Pi(I)$ such that the black boundary nodes belonging to $I$  are {\it covered} by $\Pi,$  and white nodes belonging to $I$   are {\it not covered} by $\Pi$. 
\end{definition}

\begin{theorem} \textup{\cite{L3}} \label{l2}
For a Lam model $N(\Gamma, \omega)$ the collection of boundary measurements $\Delta_{I}^d$ considered as a set of Plücker coordinates defines a point in  $\mathrm{Gr}_{\geq 0}(k(\Gamma), m)$. In particular, the boundary measurements $\Delta_{I}^d$ satisfy the Plücker relations.

Moreover, for each point  $X\in \mathrm{Gr}_{\geq 0}(k, m)$, there is a  Lam model $N(\Gamma, \omega)$ associated with it such that $\Delta_I(X)=\Delta^d_I.$
\end{theorem}
The second statement of Theorem \ref{l2} naturally leads to the formulation of the black box problem in the Postnikov theory.

 \begin{problem} \label{bl-box-post}
Consider a Lam model $N(\Gamma, \omega)$ on a given graph $\Gamma$ associated with a given point $X\in \mathrm{Gr}_{\geq 0}(k, m)$. It is required to recover a weight function $\omega$  of $N(\Gamma, \omega)$ using Plücker coordinates of $X.$ 
\end{problem}   
The possibility of solving Problem \ref{bl-box-post} is revealed by the following result.

\begin{definition}
Consider a Lam model $N(\Gamma, \omega)$. Then, its  \textbf{oriented median graph} $\Gamma_M$ is defined as follows:
    \begin{itemize}
        \item Boundary nodes of $\Gamma_M$ coincide with boundary nodes of $N(\Gamma, \omega)$;
        \item Inner nodes  of  $\Gamma_M$   are the midpoints of the  edges  connecting 
   inner nodes of $\Gamma$.
    \end{itemize}
    Two nodes of  $\Gamma_M$ are connected by an edge if the edges of the original graph $\Gamma$ are adjacent. Edges of $\Gamma_M$ are oriented as follows:
     \begin{itemize}
        \item Clockwise around each inner white node of $\Gamma$;
        \item Counterclockwise around  each inner black node of $\Gamma$.
    \end{itemize}
    Since all inner nodes of $\Gamma_M$ have degree four, we can correctly define the \textbf{strands} of $\Gamma_M$ as  the oriented paths which always go straight through any degree four node. 


\end{definition} 

\begin{definition} \label{about-min-model}
     A Lam model $N(\Gamma, \omega)$ is called  \textbf{minimal}  if it satisfies the following condition:
    \begin{itemize}
        \item Each strand does not intersect itself;
        \item There are no loops apart from loops attached to isolated boundary nodes;
        \item Strands do not form oriented lenses, see Fig. \ref{fig:forbid}.
    \end{itemize}
   
\end{definition}

\begin{figure}[H]
     \centering
     \includegraphics[scale=0.15]{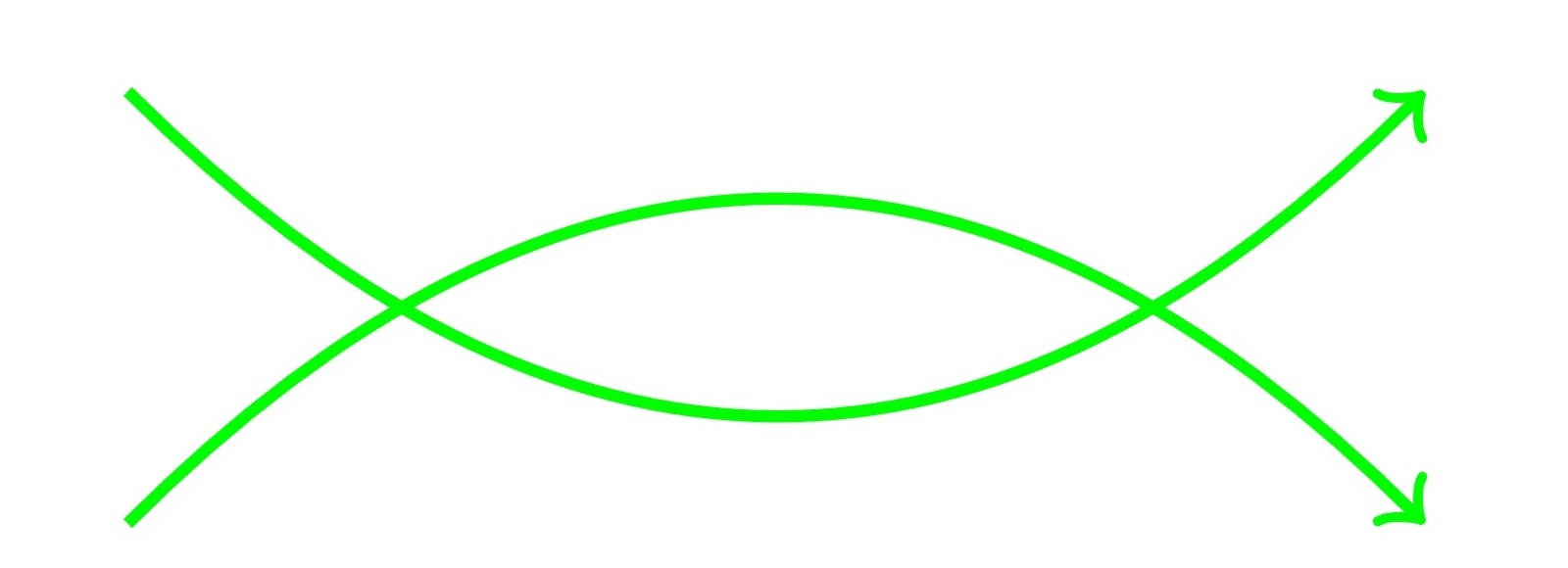}
     \caption{A forbidden oriented lens} \label{fig:forbid}
     \end{figure}
\begin{theorem} \textup{\cite{L3},\cite{Pos}} \label{thblbopost}
Consider a minimal Lam model $N(\Gamma, \omega)$ on a given graph $\Gamma$ associated with a given point $X\in \mathrm{Gr}_{\geq 0}(k, m)$. Then,  a weight function $\omega$ is uniquely  recovered   $($see Fig. \ref{fig:gau}$)$  by the  Plücker coordinates of $X,$ up to the gauge transformations.     
\end{theorem}

\begin{figure}[H]
     \centering
     \includegraphics[scale=0.3]{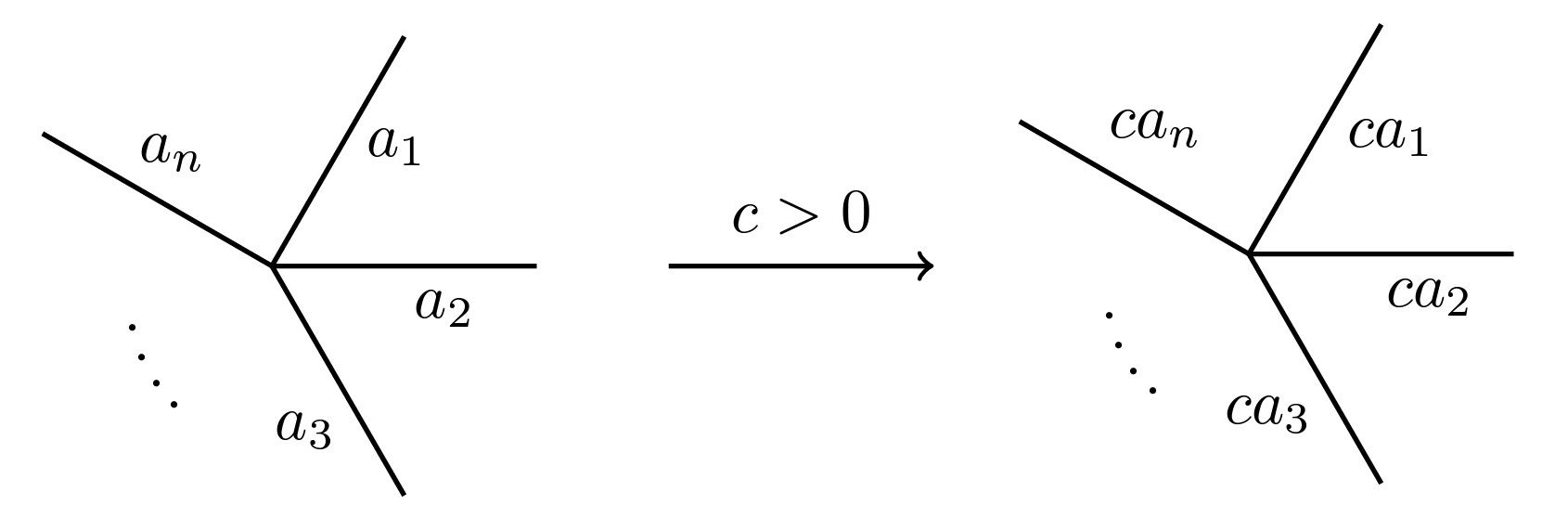}
     \caption{A gauge transformation in an inner node} \label{fig:gau}
     \end{figure}
     
     Analogous to electrical network theory,  the proofs of Theorem \ref{thblbopost}  provide  algorithms to  solve  Problem \ref{bl-box-post}. We will mention a few of them:
\begin{itemize}
    \item The recursive bridge-lollipop decomposition   \cite{L3};   
    \item Solutions to Problem \ref{bl-box-post} for  the special Le-diagram models \cite{Pos}, \cite{T};
    \item Solution to Problem \ref{bl-box-post} by the generalized chamber ansatz \cite{MSp}.
\end{itemize}
 We shall describe in detail the last approach, since it is very important for our future goal.
\begin{theorem} \textup{\cite{MSp},\cite{Sc}} \label{labelrule}
   Consider a minimal  Lam model $N(\Gamma, \omega)$. Let us label each of its faces $F$ by $I(F) \subset [n]$ according to the following the \textbf{Scott rule}: if a face $F$ lies on the left-hand side of the oriented strand from a source numbered $i$, we have that $i \in I(F)$. This labeling rule ensures that each face is labeled by $k(\Gamma)$ distinct indices.  
\end{theorem}

\begin{definition} \textup{\cite{MSp}}
   Consider a matrix $A \in \mathrm{Mat}_{k \times m}(\mathbb{R})$ and denote by $A_i, \ i \in \{1, \dots, m\}$ its columns. We define the \textbf{twist} $\tau(A)  \in \mathrm{Mat}_{k \times m}(\mathbb{R})$  as follows:
\begin{equation*}
    \begin{cases}
  (\tau (A)_i, A_i)=1,\, \text{if}\,\, i = j    \\
   (\tau (A)_i, A_j)=0,\, \text{if}\,\, A_j \notin  \text{span}(A_i, A_{i+1}, \dots, A_{j-1}),   \\
      \end{cases}
\end{equation*}
    
      where $\tau (A)_i,  i\in \{1, \dots, m\} $ are the columns of $\tau(A)$;  $(\cdot, \cdot)$ is the standard scalar product; and the  index operation $+$ is taken  modulo $m$.
\end{definition}
\begin{theorem} \textup{\cite{MSp}} \label{thabouttwit}
   Consider a minimal Lam model $N(\Gamma, \omega)$ on a given graph $\Gamma$ associated with a point $X \in \mathrm{Gr}_{\geq 0}(k, m)$ defined by a matrix $A.$ Then, a weight function  $\omega$ can be found  as follows:
   \begin{itemize}
       \item The first step is to label each face $F$ of $\Gamma$ by $I(F)$ according to the rule from Theorem \ref{labelrule};
       \item The second step is to find  the twist $\tau(A)$ of $A;$
       \item Finally we calculate the weight function $\omega$ according to the rule:
       \begin{itemize}
           \item If an edge $e$ connects two inner nodes, then 
           \begin{equation} \label{tw1}
               \omega(e)=\frac{1}{\Delta_{I(F_1)}\bigr(\tau(A)\bigl)\Delta_{I(F_2)}\bigr(\tau(A)\bigl)},
           \end{equation}
           where the edge $e$ is shared by faces $F_1$ and $F_2$; 
           \item If an edge $e$ connects an inner node with a boundary node, then
           \begin{equation} \label{tw2}
               \omega(e)=\frac{1}{\Delta_{I(F)}\bigr(\tau(A)\bigl)},
           \end{equation}
            
           where $F$ is 
           \begin{itemize}
               \item The closest counterclockwise which shares the edge $e,$ if $e$ connects a white boundary node;
               \item The closest clockwise which shares the edge $e,$ if $e$ connects a black boundary node.
           \end{itemize}
           
      \end{itemize}
   \end{itemize}
\end{theorem}
\begin{remark} \label{abouttwist}
  For our future goals, we provide the general idea of the proof of  Theorem \ref{thabouttwit} following \cite{MSp}. In fact, each face $F$ of a minimal Lam model $N(\Gamma, \omega)$ uniquely defines a \textbf{minimal} (under the partial order on a set of all dimers with a fixed boundary)  dimer $\Pi_F$, whose boundary  can be calculated   according to the rule from Theorem \ref{labelrule} i.e. $\Pi_F\in \Pi\bigr(I(F)\bigl)$. Using weights $\mathrm{wt}(\Pi_F)$ of these dimers we can  recover the weight function $\omega$ up to the gauge transformations  according to the following rule (where $F_1, F_2, F$ are defined as in Theorem \ref{thabouttwit}):
  \begin{itemize}
           \item If an edge $e$ connects two inner nodes then 
           \begin{equation*}
               \omega(e)=\mathrm{wt}(\Pi_{F_1})\mathrm{wt}(\Pi_{F_2}),
           \end{equation*}
           \item If an edge $e$ connects an inner node to a boundary node, then
           \begin{equation*} 
               \omega(e)=\mathrm{wt}(\Pi_{F}).
           \end{equation*}
           \end{itemize}
           Finally, it remains to prove that the following identity holds:
           $$\Delta_{I(F)}\bigr(\tau(A)\bigl)\mathrm{wt}(\Pi_{F})=1.$$
\end{remark}
\subsection{The Lam embedding and its parametrization} \label{thelamembsec}
We now present the cornerstone result that will enable us to apply the Postnikov theory to solve Problem \ref{bl-box}.

\begin{theorem} \textup{\cite{BGKT}} \label{th: main_gr}
 Consider an  electrical network  $\mathcal{E}(G, w) \in E_n$ with a response matrix $M_R(\mathcal{E})=(x_{ij}).$ Then there is an injective map $\mathcal{L}: E_n\to \mathrm{Gr}_{\geq 0}(n-1,2n)$ such that $\mathcal{L}(\mathcal{E})$ is a row space of the $(n\times 2n)$ matrix
\begin{equation*}
\Omega(\mathcal{E})=\left(
\begin{array}{cccccccc}
x_{11} & 1 & -x_{12} & 0 & x_{13} & 0 & \cdots & (-1)^n\\
-x_{21} & 1 & x_{22} & 1 & -x_{23} & 0 & \cdots & 0 \\
x_{31} & 0 & -x_{32} & 1 & x_{33} & 1 & \cdots & 0 \\
\vdots & \vdots &  \vdots &   \vdots &  \vdots & \vdots & \ddots & \vdots 
\end{array}
\right)    
\end{equation*}
In particular we have:
\begin{itemize}
    \item  The dimension of the row space of $\Omega(\mathcal{E})$ is equal to $n-1$;
    \item Each $n-1 \times n-1$ minors     of  $\Omega(\mathcal{E})$  is non-negative;
   \item   Plücker coordinates of the point of $\mathcal{L}(\mathcal{E})$ correspond  to $n-1 \times n-1$ minors  of the matrix $\Omega'(\mathcal{E})$ obtained from $\Omega(\mathcal{E})$  by deleting the last row.
\end{itemize}
\end{theorem}
\begin{remark}
    The non-negativity of the Plücker coordinates of  points $\mathcal{L}(\mathcal{E})$ is essentially equivalent to the last   property of response matrices described in  Theorem   \ref{aboutresp} $($see \cite{BGK} for more details$)$.
\end{remark}
Due to Theorem \ref{l2}   a point $\mathcal{L}(\mathcal{E})$ associates with a Lam model, which can be canonically constructed with a network $\mathcal{E}$ by the \textbf{generalized Temperley trick}.  A model obtained by this trick is essentially constructed by combining a graph $G$ and its dual $G^{*}$. The nodes of both $G$ and $G^{*}$ are colored black. An additional set of nodes introduced at the intersections of edges from $G$ and $G^{*}$ are colored white. Edges originating from $G$ inherit a weight equal to the conductivity of the corresponding edge in $G$, while all new  edges are assigned a weight of $1$. More formally:
\begin{definition} \textup{\cite{L}}  \label{temp_gen}
Given an electrical network $\mathcal{E}(G, w)\in E_n$, let us construct its corresponding Lam model $N(\mathcal{E^T}, \boldsymbol \omega)$ defined on a graph $\mathcal{E^T}$ with weight function $ \boldsymbol \omega$. The nodes of $\mathcal{E^T}$ are defined as follows:
\begin{itemize}
    \item If $G$ has boundary nodes $\{1,  2, \dots,  n\}$, then $\mathcal{E^T}$ will have white boundary nodes $\{1, 2, \dots , 2n\}$, where boundary node $ i$ is identified with $2i-1$ and node $2i$ can be identified with the additional node  lied between $i$ and $i+1$;
    \item We have a black inner node $b_v$ for each inner node $v$ of $G$; a black inner node $b_F$ for each interior face $F$ of $G$; a white inner node $w_e$ placed at the midpoint of each interior edge $e$ of $G;$ 
    \item  For each boundary node $i$ of $G$, we have a black inner node  $b_i$.
\end{itemize}  
The edges of $N(\mathcal{E^T}, \boldsymbol \omega)$ are defined as follows: 
\begin{itemize}
\item  If $v$ is a node of an edge $e$ in $G$, then $b_v$ and $w_e$ are joined, and the weight $\boldsymbol \omega(e)$ of this edge is equal to the weight $w(e)$ of $e$ in $G$; 
\item  If $e$ borders $F$, then $w_e$ is joined to $b_F$ by an edge with weight $1$; 
\item  The node $b_i$ is joined by an edge with weight $1$ to the boundary node $2i - 1$ in $N(\mathcal{E^T}, \boldsymbol \omega)$, and $b_i$
is also joined by an edge with weight $1$ to $w_e$ for any edge $e$ incident
to $ i$ in $G$; 
\item  Even boundary nodes $2i$ in $N(\mathcal{E^T}, \boldsymbol \omega)$ are joined by an edge with weight $1$ to the face vertex $w_F$ of the face $F$ that they lie in.
\end{itemize}

\end{definition}

\begin{figure}[H]
     \hspace*{-14mm}
     \includegraphics[scale=0.11]{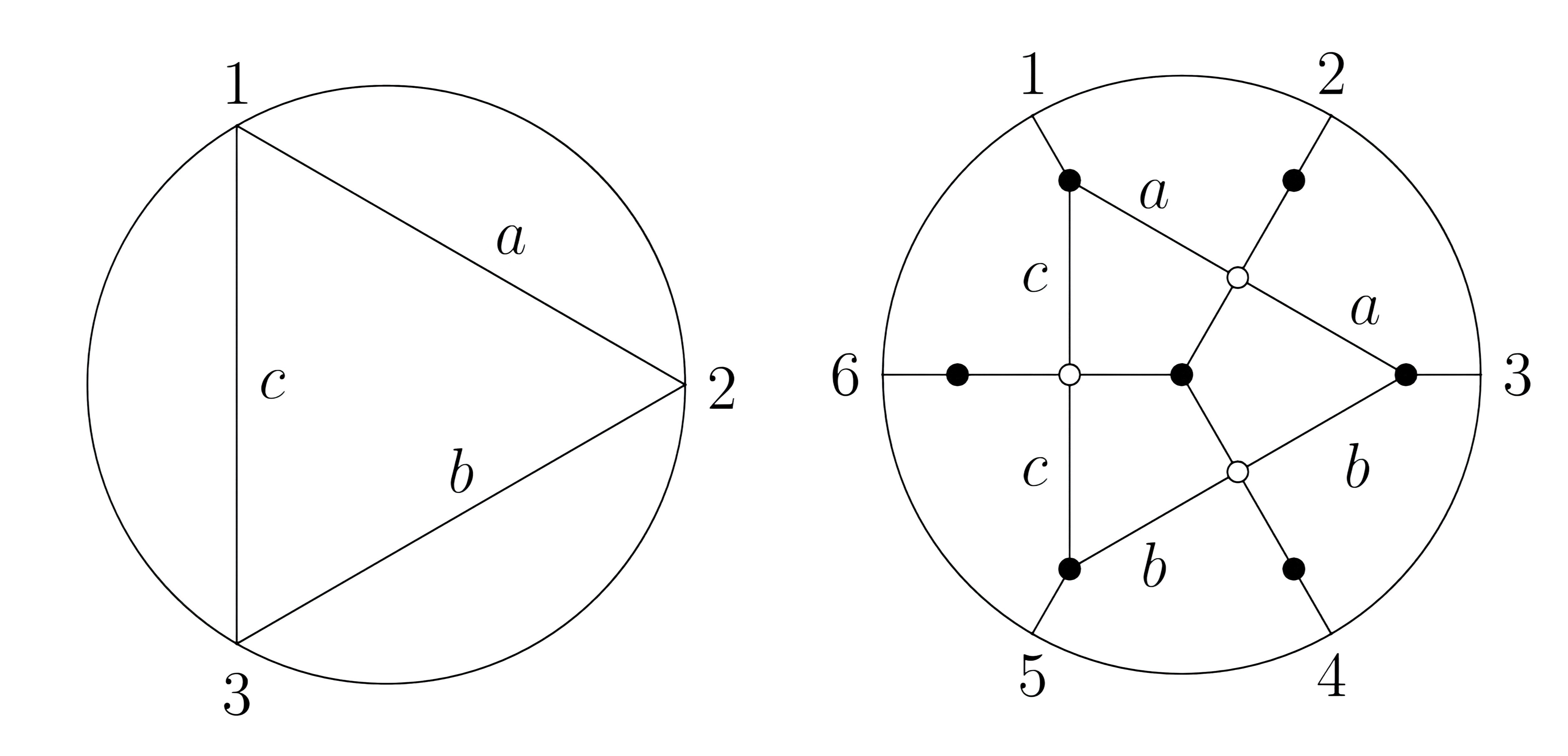}
     \caption{The triangle-shape network and its Lam model $N(\mathcal{E^T}, \boldsymbol \omega)$} 
     \end{figure}
 
\begin{theorem} \label{connection}
\hskip 1pt
\begin{itemize}
    \item A Lam model $N(\mathcal{E^T}, \boldsymbol \omega)$ and $\Omega(\mathcal{E})$ define the same point of $ \mathrm{Gr}_{\geq 0}(n-1,2n),$ it was proved in  \cite{BGKT};
    \item Let us suppose that  a network $\mathcal{E} $ is  minimal and connected, then a Lam model $N(\mathcal{E^T}, \boldsymbol \omega)$ is also minimal.
\end{itemize}
\end{theorem}
\begin{figure}[h!]
    \centering
    \includegraphics[width=0.8\textwidth]{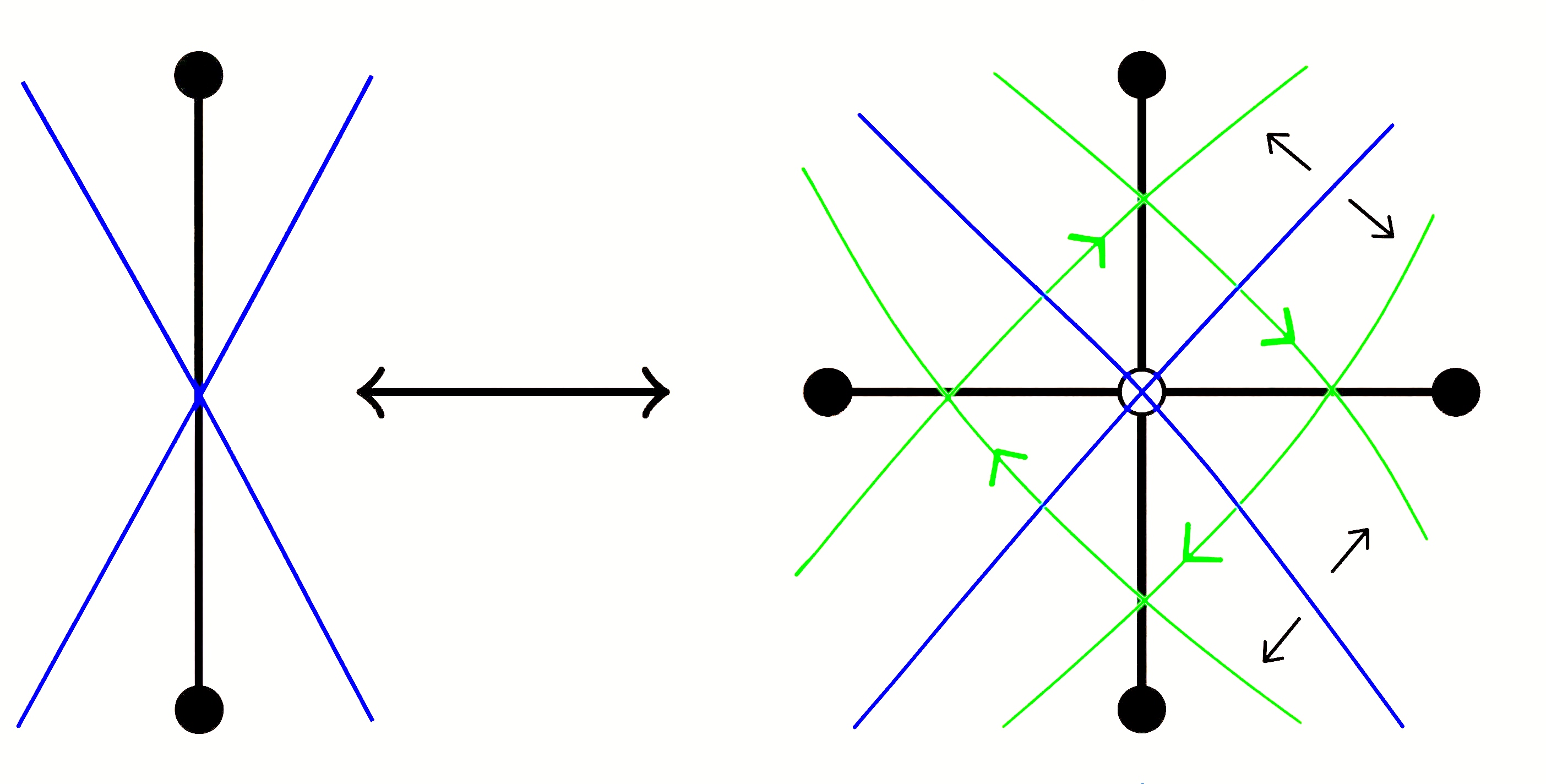}
    \caption{Strands of $G_M$ and $\Gamma_M$ around a midpoint of an electrical network edge}
    \label{strand-fig}
\end{figure}
\begin{proof}
    To prove the second statement, it suffices to verify that $N(\mathcal{E^T}, \boldsymbol \omega)$ satisfies  conditions in Definition  \ref{about-min-model}. Indeed, the strands of a  median graph $G_M$ of a Lam model $N(\mathcal{E^T}, \boldsymbol \omega)$ are obtained by the parallel transitions (see Fig. \ref{strand-fig}) of strands of a median graph $\Gamma_M$ of a network $\mathcal{E}$. Since  $\mathcal{E}$ is minimal, any two    strands of $\Gamma_M$ intersect at most once and each strand does not intersect itself.  Therefore, using the above argument  we conclude that:
    \begin{itemize}
        \item Any two  strands of $G_M$ can intersect at most once in a midpoint of an inner edge of $N(\mathcal{E^T}, \boldsymbol \omega)$;
        \item No strand of $G_M$  intersect itself.
    \end{itemize}
      Thus, any two strands might  intersect:
     \begin{itemize}
     \item Exactly once in a  midpoint of an inner edge;
        \item Exactly once in a  midpoint of an inner edge and exactly once  in a boundary node. Since all boundary nodes are white and a   midpoint of an inner edge is connected with a white inner node, it is easy to check that this case does not lead to the forbidden lens, see Fig. \ref{fig:forbid}; 
        \item Exactly once in a  midpoint of an inner edge and two times   in a boundary nodes. According to  the reasoning    above, these  intersections are not possible, since they contradict to the orientation rule of $G_M;$ 
        \item  Two times   in a boundary nodes.  This case also does not lead to the forbidden lens.
    \end{itemize}
    Therefore, any  two strands do not form a forbidden lens, it ends the proof. 
\end{proof}

\section{The solution to the black  problem} \label{secan}
\subsection{Conductivity recovering via the generalized chamber ansatz}
The generalized Temperley trick and Theorems \ref{thblbopost} and \ref{connection} allow us to reduce Problem \ref{bl-box} to a particular case of Problem \ref{bl-box-post}:
\begin{proposition}
    Consider a minimal electrical network $\mathcal{E}(G, w)$. Then, the recovery of a  conductivity function $w$ by a response matrix $M_R(\mathcal{E})$  is equivalent to the recovery of a  weight function $\omega$ of a Lam model $N(\mathcal{E^T}, \boldsymbol \omega)$ by a  matrix $\Omega(\mathcal{E})$.
\end{proposition}
As a consequence of Theorems \ref{thblbopost} and \ref{thabouttwit}, we obtain that, for any minimal network, we can reconstruct a weight function $\boldsymbol\omega'$ for a Lam model $N(\mathcal{E^T}, \boldsymbol \omega)$. This  function is equal to the original weight function $\boldsymbol\omega$ from Definition \ref{temp_gen} up to gauge transformations. However, explicitly determining the sequence of gauge transformations necessary to transform $\boldsymbol\omega'$ into $\boldsymbol\omega$ is often difficult. Fortunately, there are  specific  coordinates on $N(\mathcal{E^T}, \boldsymbol \omega)$ that remain invariant under such gauge transformations:

        

\begin{lemma} \textup{\cite{Pos}} \label{lem:post-cord}
      We define the weight of each face $F$ of  $N(\mathcal{E^T}, \boldsymbol \omega)$ as follows:
    $$O(F)=\frac{\prod \limits_{e \in F,\ e=wb } \boldsymbol\omega(e)}{\prod \limits_{e \in F,\ e=bw} \boldsymbol\omega(e)},$$
    where $F$ is considered to be clockwise edge oriented;  the numerator is the product taken over all from white to black  edges of $F$ and the denominator is the product taken over all from black to white  edges of $F$.

   The face weight $O(F)$ does not depend on the gauge transformations.
\end{lemma}

\begin{figure}[h!]
     \centering
     \includegraphics[scale=0.1]{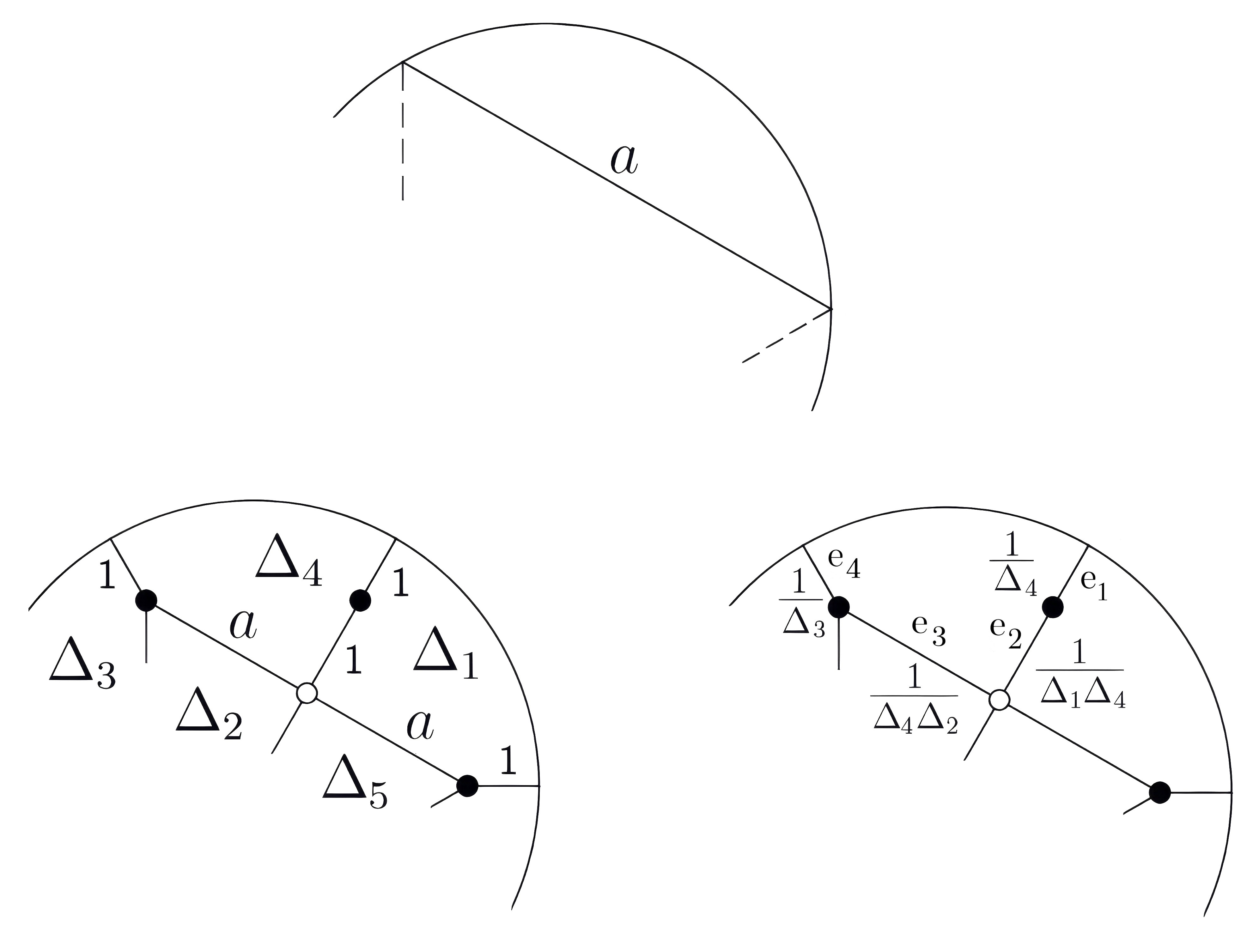}
     \caption{A bridge and a Lam model $N(\mathcal{E^T}, \boldsymbol \omega)$ around it} \label{figch1}
     \end{figure}
     \begin{figure}[h!]
     \centering
     \includegraphics[scale=0.25]{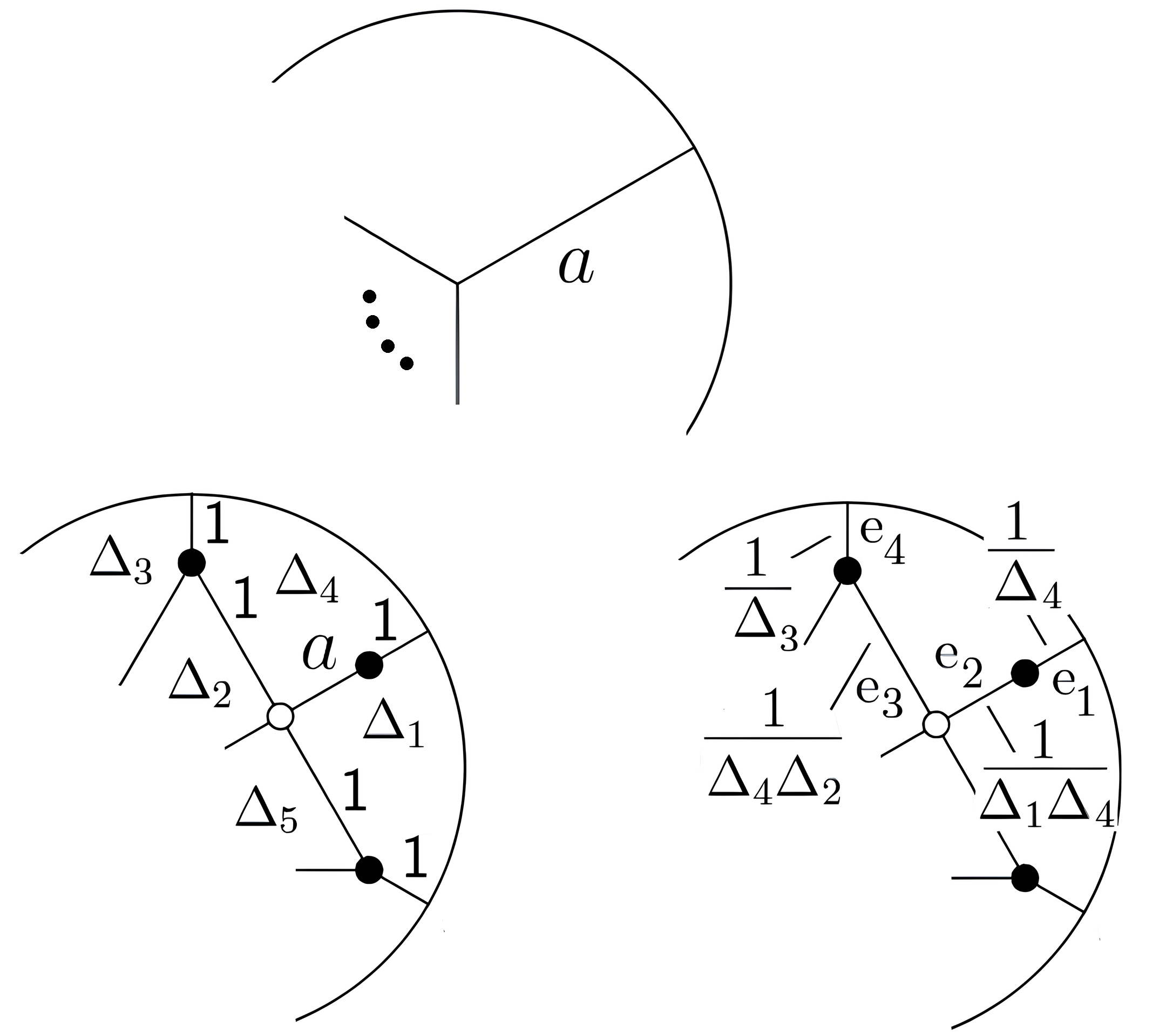}
     \caption{A spike and a Lam model $N(\mathcal{E^T}, \boldsymbol \omega)$ around it} \label{figch2}
     \end{figure}
     
\begin{figure}[h!]
     \centering
     \includegraphics[scale=0.15]{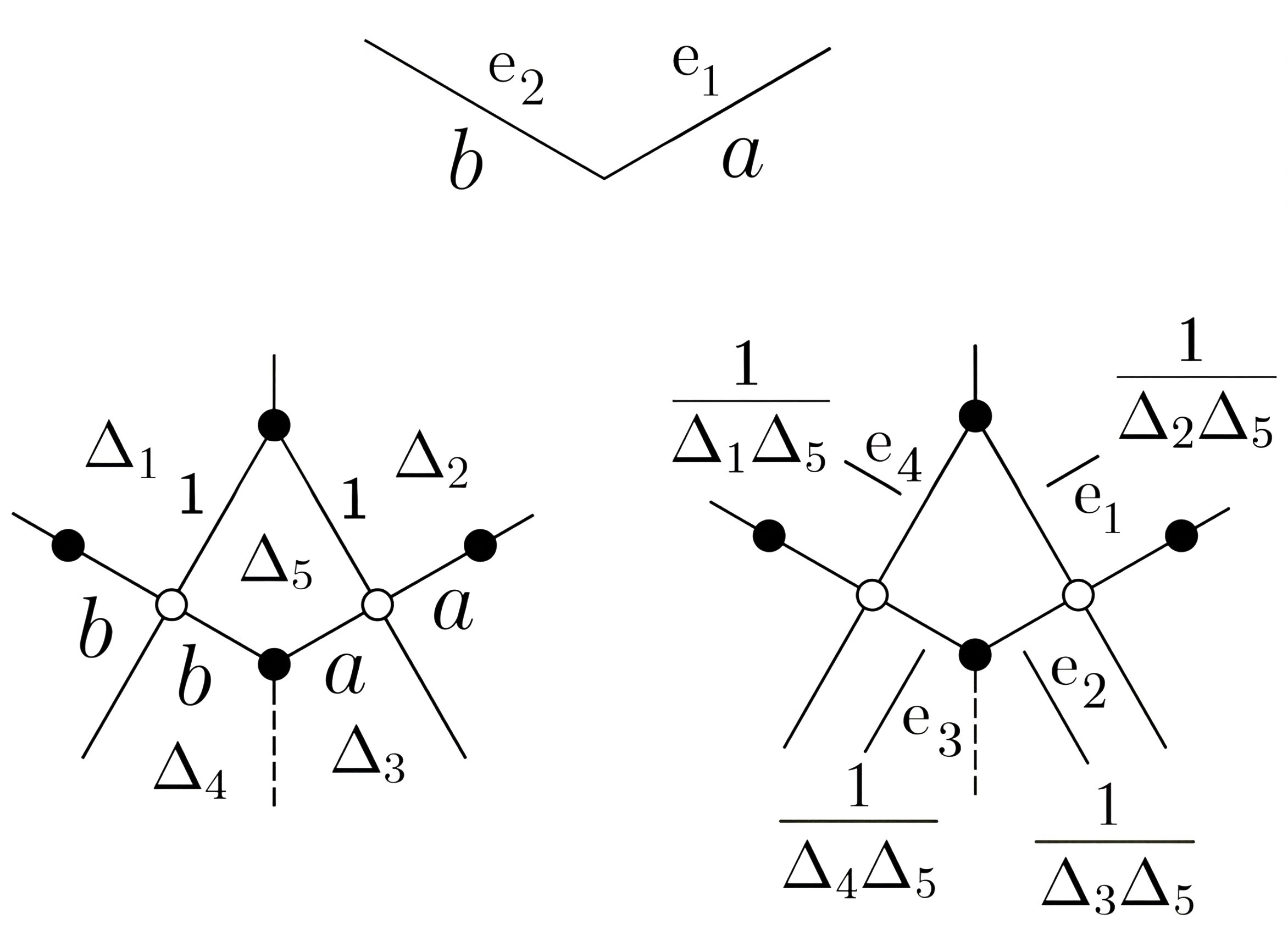}
     \caption{Two adjacent edges  and a Lam model $N(\mathcal{E^T}, \boldsymbol \omega)$ around them} \label{figch3}
     \end{figure}
    Using the coordinates introduced in Lemma \ref{lem:post-cord}, we can formulate a complete analog of the chamber ansatz (see \cite{BFZ}, \cite{MSp}) for the conductivities of electrical networks:
\begin{lemma} \label{mainlemma}
      Consider a minimal electrical network $\mathcal{E}(G, w)$ and its edge $e$. Then, the following holds:
      \begin{itemize}
           \item If an edge $e$ is a \textbf{bridge} $($i.e. an edge  which connects a boundary node numbered by $i$ with a boundary node numbered by $i+1)$   its conductivity can be recovered by the formula
           \begin{equation} \label{ch1}
           w(e)=\dfrac{\Delta_1\Delta_3}{\Delta_2\Delta_4}
           \end{equation}
           or, alternatively, by the formula
            \begin{equation} \label{ch1'}
           w(e)=\dfrac{\Delta_1}{\Delta_5};    
           \end{equation}
          \item   If an edge $e$ is a \textbf{spike} $($i.e. an edge  which connects a boundary node  with an unique inner node$)$  can be recovered by the formula
          \begin{equation} \label{ch2}
           w(e)=\dfrac{\Delta_2\Delta_4}{\Delta_1\Delta_3}   
           \end{equation}
           or, alternatively, by the formula
            \begin{equation} \label{ch2'}
           w(e)=\dfrac{\Delta_5}{\Delta_1};    
           \end{equation}
          
          \item Let us consider a face $F$ and its two  adjacent edges $e_1$ and $e_2$ $($see Fig. \ref{figch3}$)$, then the  ratio of their conductivities is equal to 
            \begin{equation} \label{ch3}
          \dfrac{w(e_1)}{w(e_2)}=\dfrac{\Delta_2\Delta_4}{\Delta_1\Delta_3},   
           \end{equation}
      \end{itemize}
      where $F_i$ are the faces labeled by $\Delta_i$ as it is shown in Fig. \ref{figch1}, Fig. \ref{figch2} and Fig. \ref{figch3}; and by $\Delta_i$ we denote $$\Delta_i:=\Delta_{I(F_i)}\Bigr(\tau \bigr(\Omega'(\mathcal{E})   \bigl)\Bigl).$$
\end{lemma}
\begin{proof}
Let us prove Formula \eqref{ch1}. Indeed, let us consider a face $F_1$, then $$O(F_1)=\dfrac{\boldsymbol\omega(e_1)\boldsymbol\omega(e_3)}{\boldsymbol\omega(e_2)\boldsymbol\omega(e_4)}=a.$$  
On the other hand, using Theorem \ref{thabouttwit} we have that  
$$O(F_1)=\dfrac{\boldsymbol\omega'(e_1)\boldsymbol\omega'(e_3)}{\boldsymbol\omega'(e_2)\boldsymbol\omega'(e_4)}=\dfrac{\Delta_1\Delta_3}{\Delta_2\Delta_4},$$
where up to gauge transformations $\boldsymbol\omega'=\boldsymbol\omega$   and 
$$\boldsymbol\omega'(e_1)=\dfrac{1}{\Delta_4}, \ \boldsymbol\omega'(e_2)=\dfrac{1}{\Delta_1\Delta_4}, \
\boldsymbol\omega'(e_3)=\dfrac{1}{\Delta_4\Delta_2}, \
\boldsymbol\omega'(e_4)=\dfrac{1}{\Delta_3}.$$
Since the value of $O(F_1)$ is independent on the gauge transformations, we obtain Formula \eqref{ch1}. The remaining formulas can be derived in a similar way, which completes the proof.
\end{proof}

Based on Lemma \ref{mainlemma}, we propose the following algorithm to solve Problem \ref{bl-box}.
\begin{algorithm} \label{mainalg}
\hskip 1pt
\begin{itemize}
    \item Given a response matrix $M_R(\mathcal{E})$ of a minimal network $\mathcal{E}$ construct the matrix $\Omega'(\mathcal{E})$;
    \item  Calculate $\tau\bigl(\Omega'(\mathcal{E})\bigr);$
    \item  Label each face $F$ of $N(\mathcal{E^T}, \boldsymbol \omega)$ by $n-1 \times n-1$ minors $\Delta_{F(I)}$ of a $\tau\bigl(\Omega'(\mathcal{E})\bigr)$;
    \item Using Formulas \eqref{ch1}--\eqref{ch1'}, \eqref{ch2}--\eqref{ch2'} and \eqref{ch3} recover a conductivity function $w.$
\end{itemize}
\end{algorithm}
\begin{proof}
    The correctness of the proposed algorithm is derived from the following statements:
\begin{itemize}
    \item Each minimal electrical network has at least one bridge or spike (see \cite{CIW}), therefore we can extend bridges and spikes conductivities to conductivities of remaining  edges by Formula \eqref{ch3}; 
    \item The uniqueness of this extension defined    by Formula \eqref{ch3} follows from Theorem \ref{thblbopost} and the second part of  Theorem \ref{connection}.
\end{itemize}
\end{proof}
\begin{figure}[h!]
     \centering
     \includegraphics[scale=0.55]{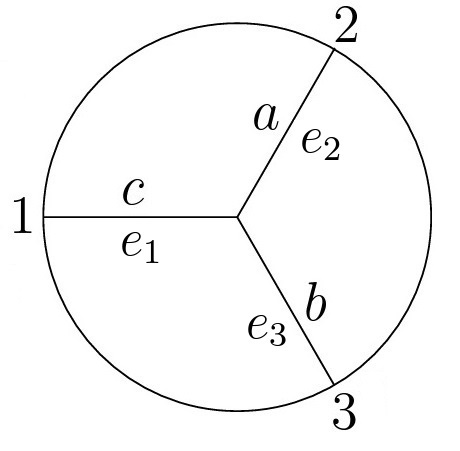}
     \caption{The star-shape network  $\mathcal{E}$} \label{ex1}
     \end{figure}
\begin{remark}
   As demonstrated by Equations \eqref{ch1}--\eqref{ch3}, the conductivity function $w$ can be recovered in multiple ways. Furthermore, every minimal electrical network contains at least three bridges and spikes in total, see \cite{CIW}.   Consequently, the process of recovering the conductivity function $w$ leads to a multitude of non-trivial identities, which of course follow from   the special form of the  weight function $\boldsymbol \omega$ of $N(\mathcal{E^T}, \boldsymbol \omega)$. These identities  may  hint at the existence of   integrable structures (see  \cite{BS}) related to electrical networks.
\end{remark}
 Let us consider a couple of examples.
\begin{example}
   Consider a network  $\mathcal{E}$ as it is shown in Fig. \ref{ex1}. Its response matrix has the form:
   \begin{equation*} 
		M_R(\mathcal{E})= \left(
		\begin{array}{cccccccc}
			\frac{ca+cb}{a+b+c}  & -\frac{ca}{a+b+c} &   -\frac{cb}{a+b+c} \\
			-\frac{ca}{a+b+c} &  \frac{ca+ab}{a+b+c} &   -\frac{ab}{a+b+c} \\
		-\frac{cb}{a+b+c} &  -\frac{ab}{a+b+c} &   -\frac{cb+ab}{a+b+c} 
		\end{array}
		\right).
	\end{equation*} 
Then, $\Omega_3(\mathcal{E})$ and $\Omega'_3(\mathcal{E})$ equal to:

\begin{equation*}
\Omega_3(\mathcal{E})= \left(
		\begin{array}{cccccccc}
			\frac{ca+cb}{a+b+c} & 1 & \frac{ca}{a+b+c} & 0 &  \frac{-cb}{a+b+c} & -1\\
			\frac{ca}{a+b+c} & 1 & \frac{ca+ab}{a+b+c} & 1 &  \frac{ab}{a+b+c} & 0 \\
		\frac{-cb}{a+b+c} & 0 & \frac{ab}{a+b+c} & 1 &  \frac{cb+ab}{a+b+c} & 1 
		\end{array}
		\right), \
		\Omega'_3(\mathcal{E})= \left(
		\begin{array}{cccccccc}
			\frac{ca+cb}{a+b+c} & 1 & \frac{ca}{a+b+c} & 0 &  \frac{-cb}{a+b+c} & -1\\
			\frac{ca}{a+b+c} & 1 & \frac{ca+ab}{a+b+c} & 1 &  \frac{ab}{a+b+c} & 0
		
		\end{array}
		\right).
	\end{equation*}  

By  direct  computation we can verify that all $2 \times 2$ minors of $\Omega'_3(\mathcal{E})$ do not vanish, therefore
\begin{equation*} 
		\tau\bigl(\Omega'_3(\mathcal{E})\bigr)= \left(
		\begin{array}{cccccccc}
			\frac{a+b+c}{bc} & \frac{c+b}{b} & \frac{a+b+c}{ac} & \frac{a}{c} &  0 & -1\\
			-\frac{a+b+c}{bc} & -\frac{c}{b} & 0  & 1 &  \frac{a+b+c}{ab} & \frac{a+b}{a}
		\end{array}
		\right).
	\end{equation*}
For instance,   $\tau(A)_6$ is defined by the following conditions: $( \tau(A)_6, A_6)=1$ and  $( \tau(A)_6, A_{6+1 \ \text{mod} \ 6})=( \tau(A)_6, A_{1})=0$ (since  $A_6 \notin \text{span}(A_1)$). 

Let us label each face of $N(\mathcal{E^T}, \boldsymbol \omega)$ according to the rule defined in Theorem \ref{labelrule}, as it is shown in Fig. \ref{ex2}. Since the edge $e_1$ is a spike, we obtain the following identities:
\begin{equation} \label{formtogrove}
    w(e_1)=\dfrac{\Delta_{26}}{\Delta_{16}}=\dfrac{\frac{a+b+c}{a}}{\frac{a+b+c}{ca}}=c,
\end{equation}
or alternatively 
\begin{equation}
    w(e_1)=\dfrac{\Delta_{56}\Delta_{46}}{\Delta_{16}\Delta_{45}}=\dfrac{\frac{a+b+c}{ab}\frac{a+b+c}{c}}{\frac{a+b+c}{ca}\frac{a+b+c}{cb}}=c.
\end{equation}
\end{example}
      \begin{figure}[h!]
     \centering
     \includegraphics[scale=0.17]{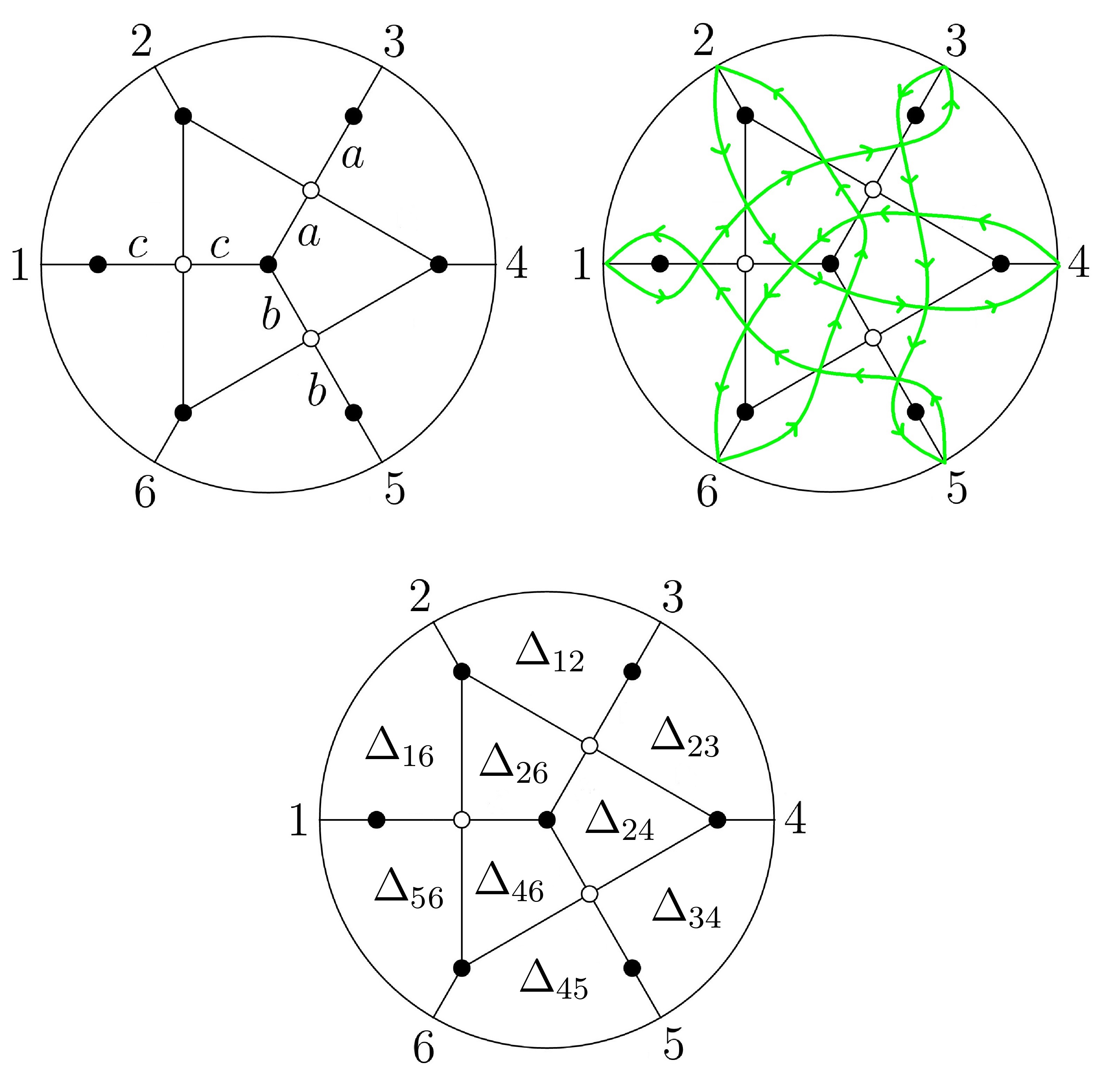}
      \caption{The Lam model  $N(\mathcal{E^T}, \boldsymbol \omega)$ and the face labeling} \label{ex2}
     \end{figure}

\begin{figure}[H]
     \hspace*{-11mm}
     \includegraphics[scale=0.18]{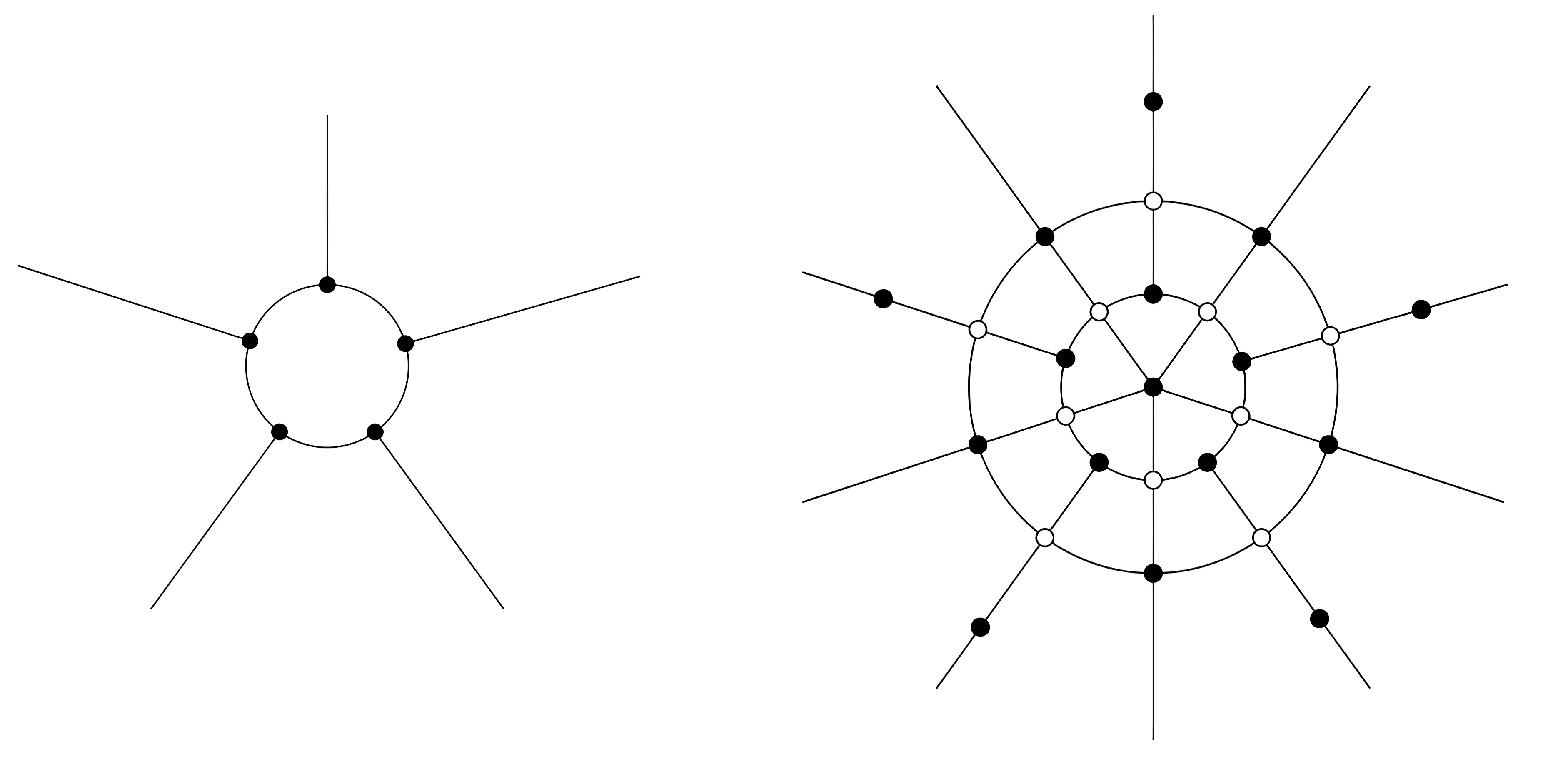}
     \caption{$G_5$ and its associated  Lam model. The boundary circles are not drawn} \label{g5gr}
     \end{figure}
\begin{example}
Let us consider a series of the general position (well-connected, see \cite{K}) minimal  electrical networks $\mathcal{E}(G_{4m+1}, w)$ defined as follows:
\begin{itemize}
    \item The boundary nodes lie on a boundary  circle of radius $m+1$, centered at $(0,0)$;
    \item The nodes of $G_{4m+1}$ are the points $(i,j)$ for integers $i$ and $j$ with $0<i\le m+1$ and $1\le j\le 4m+1$. The radial edges are the radial line segments joining $(i,j)$ to $(i+1,j)$ for each $0<i\le m$ and  $1\le j\le 4m+1$. The circular edges are the circular arcs joining $(i,j)$ to $(i,j+1)$ for each $1\le i\le m$ and each $1\le j\le 4m+1$. The graph $G_{4m+1}$ has $2m(4m+1)$ edges and $(m-1)(4m+1)$ nodes;
    \item The boundary nodes of $G_{4m+1}$ are the points $v_j=(m+1,j)$, for $j=1,\ldots,4m+1$, with the convention that $v_0=v_{4m+1}$.
\end{itemize}
The significance of these electrical networks follows from their role as duals to the networks utilized in numerical solutions to the Calderón problem, see \cite{BDV} and also  \cite[Theorem $3.7$]{GK}.

    Based on  the specific shape of $\mathcal{E}(G_{4m+1}, w)$, we can propose the following way to solve Problem \ref{bl-box} to them:
    \begin{itemize}
        \item Label each face of $N(\mathcal{E^T}, \boldsymbol \omega)$ according to the rule described in \cite{BGK};
        \item Recover spikes conductivities;
        \item Using a spike  conductivity and Formula \eqref{ch3},    extend  conductivities of the remaining edges in each \textbf{part}, see Fig. \ref{part} and  \ref{partup}:  
        \begin{equation*}
            a_1=\dfrac{\Delta_1\Delta_2}{\Delta'_1\widetilde{\Delta}_1}, \ \dfrac{a_i}{b_i}=\dfrac{\Delta_{2i-1}\Delta_{2i+1}}{\Delta'_{2i}\widetilde{\Delta}_{2i}}, \ \dfrac{b_i}{a_{i+1}}=\dfrac{\Delta'_{2i+1}\widetilde{\Delta}_{2i+1}}{\Delta_{2i}\Delta_{2i+2}}.
        \end{equation*}
    \end{itemize}
\end{example}
 \begin{figure}[H]
     \hspace*{8mm}
     \includegraphics[scale=0.4]{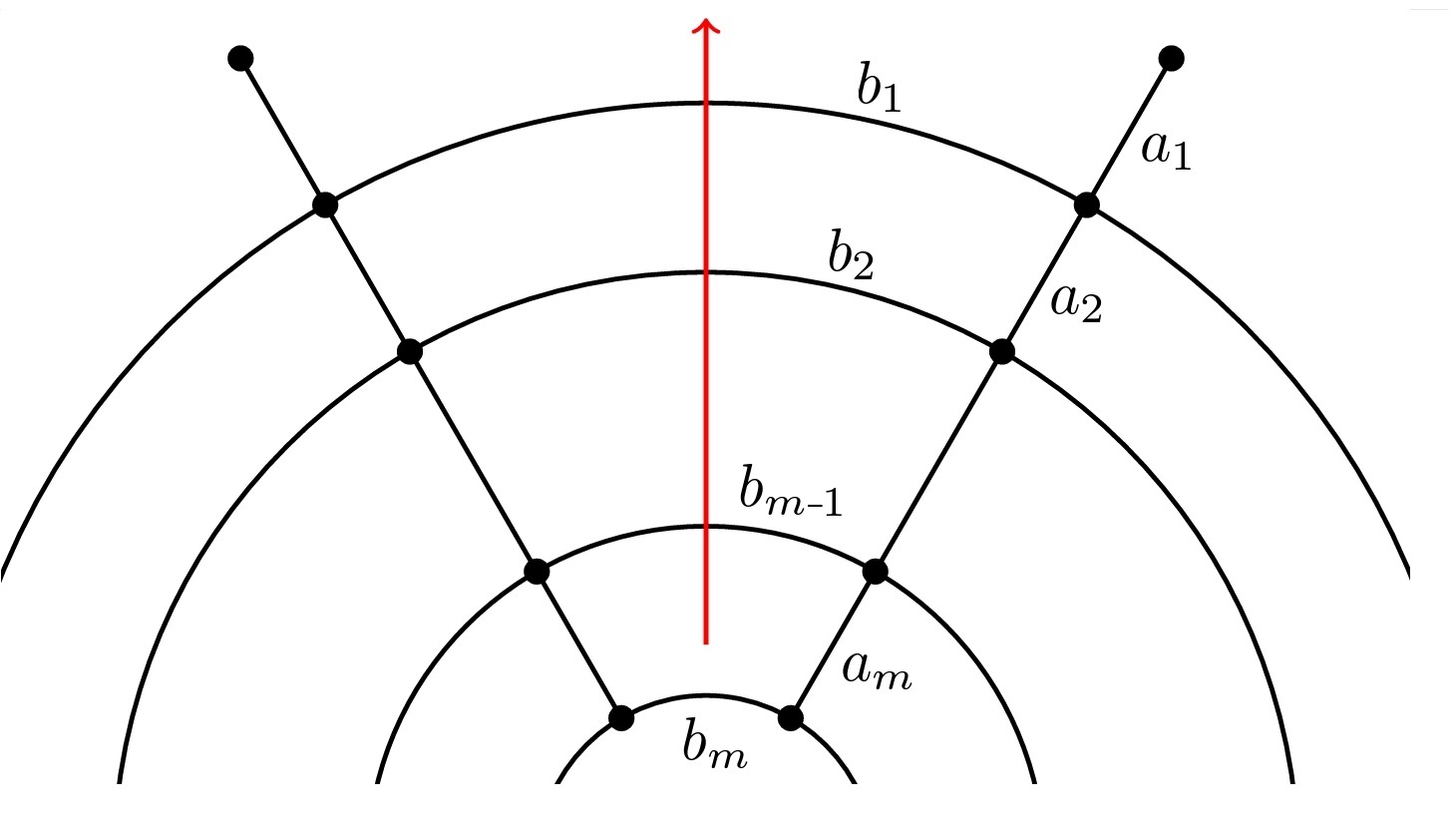}
      \caption{The part} \label{part}
     \end{figure}
\begin{figure}[H]
     \centering
     \includegraphics[scale=0.4]{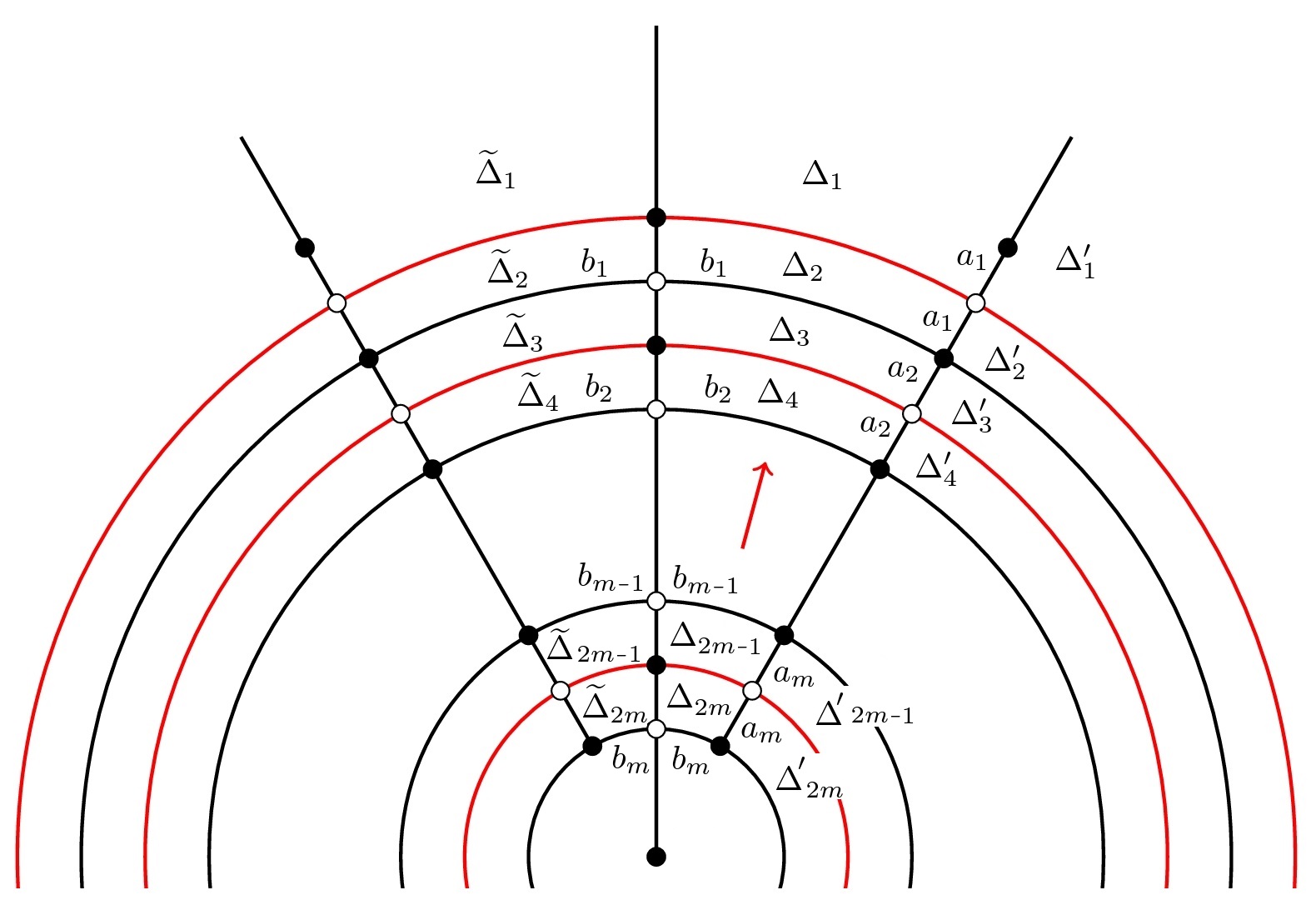}
      \caption{A Lam model across the part} \label{partup}
     \end{figure}
\subsection{Conductivity recovering via groves} \label{efsec} 
 Now, we reveal some combinatorial aspects of Algorithm \ref{mainalg}. Particularly, we demonstrate that conductivities  can be recovered by weights of the special  groves.

\begin{definition} \label{groves, ncp}
A \textbf{grove} $\mathcal{G}$ on a graph $G$ of an electrical network $\mathcal{E}(G, w)$ is a spanning forest, i.e. an acyclic subgraph of $G$ covering all nodes of $\mathcal{E}$, such that each  connected component of $\mathcal{G}$  contains at least one boundary node. 

A weight $wt(\mathcal{G})$ of the grove $\mathcal{G}$ is equal to $wt(\mathcal{G})=\prod \limits_{e \in \mathcal{G}}w(e).$ 
\end{definition}

\begin{theorem} \cite{L} \label{thgroves}
    Consider an  electrical network $\mathcal{E}(G, w)$ and its associated  Lam model $N(\mathcal{E^T}, \boldsymbol \omega).$ Then, for each $I \subset \{1,\dots, n\}, \ |I|=n-1$, there is the natural weight preserving \textbf{bijection} $\mathcal{F}_I$ between the subset of dimers $\Pi \in \Pi(I)$ on     $N(\mathcal{E^T}, \boldsymbol \omega)$ $($i.e. a subset of dimers, which do not cover boundary nodes belonging to $I$, see Definition \ref{ld}$)$   and its corresponding   subset of  groves on $\mathcal{E}(G, w)$ such that $wt(\Pi)=wt\bigl(\mathcal{F}_I(\Pi)\bigr).$ 
\end{theorem} 
Due to Theorem \ref{thgroves}, the Plücker coordinates of $\mathcal{L}(\mathcal{E})$ can be rewritten in terms of  grove weights. 
 \begin{theorem} \cite{BGKT} \label{thgroves1}
    Consider an  electrical network $\mathcal{E}(G, w) \in E_n$. Then, the following identity holds:
    \begin{equation*}
        \Delta_I\bigl(\Omega(\mathcal{E})\bigr)=\dfrac{\sum \limits_{\Pi\in \Pi(I)} wt(\Pi)}{L_{unc}}=\dfrac{\sum \limits_{\Pi\in \Pi(I)} wt\bigl(\mathcal{F}_I(\Pi)\bigr)}{L_{unc}}=\dfrac{\sum \limits_{\mathcal{G}|I} wt\bigl(\mathcal{G}\bigr)}{L_{unc}},
    \end{equation*}
    where $L_{unc}$ is the partition function of the weights of all groves with exactly $n$ connected components; and the last sum goes over  all groves, which are concordant with $I$, see \cite{L} for more details.
\end{theorem} 
 Theorem \ref{thgroves} and Remark \ref{abouttwist} also provide additional interpretations of the formulas presented in Lemma \ref{mainlemma}. 
\begin{proposition} \label{mainprop}
    Consider a minimal electrical network $\mathcal{E}(G, w)$. Then, the following Plücker coordinate equality holds:
    \begin{equation*}
        \Delta_{I(F)}\Bigr(\tau \bigr(\Omega'(\mathcal{E})   \bigl)\Bigl)=\dfrac{1}{wt(\mathcal{G}_F)},
    \end{equation*}
    where $\mathcal{G}_F:=\mathcal{F}_I(\Pi_F)$ is the \textbf{minimal} grove corresponded to the minimal dimer  $\Pi_F,$ see Remark \ref{abouttwist}.
\end{proposition}
\begin{remark} \label{remaboutboundary}
    Moreover, if a face $F$ of $N(\mathcal{E^T}, \boldsymbol \omega)$ is a boundary face (i.e. a face which is bounded by the boundary circle of $N(\mathcal{E^T}, \boldsymbol \omega)$), then $\Pi_F$ is a unique dimer such  that $\Pi \in I_F(\Pi)$, see \cite{MSp}. Therefore,   we obtain  the following Plücker coordinate equality for the boundary face:
    \begin{equation}
        \Delta_{I(F)} \bigr(\Omega'(\mathcal{E})   \bigl)=wt(\mathcal{G}_F).
    \end{equation}
\end{remark}
 \begin{proposition} \label{mainprop1}
Combining Proposition \ref{mainprop} with Lemma \ref{mainlemma}, we conclude that the conductivity function of each minimal electrical network   can be recovered by weights of the minimal groves.     
 \end{proposition}
 Let us provide  an example: 
\begin{figure}[h!]
     \hspace*{-10mm}
     \includegraphics[scale=0.1]{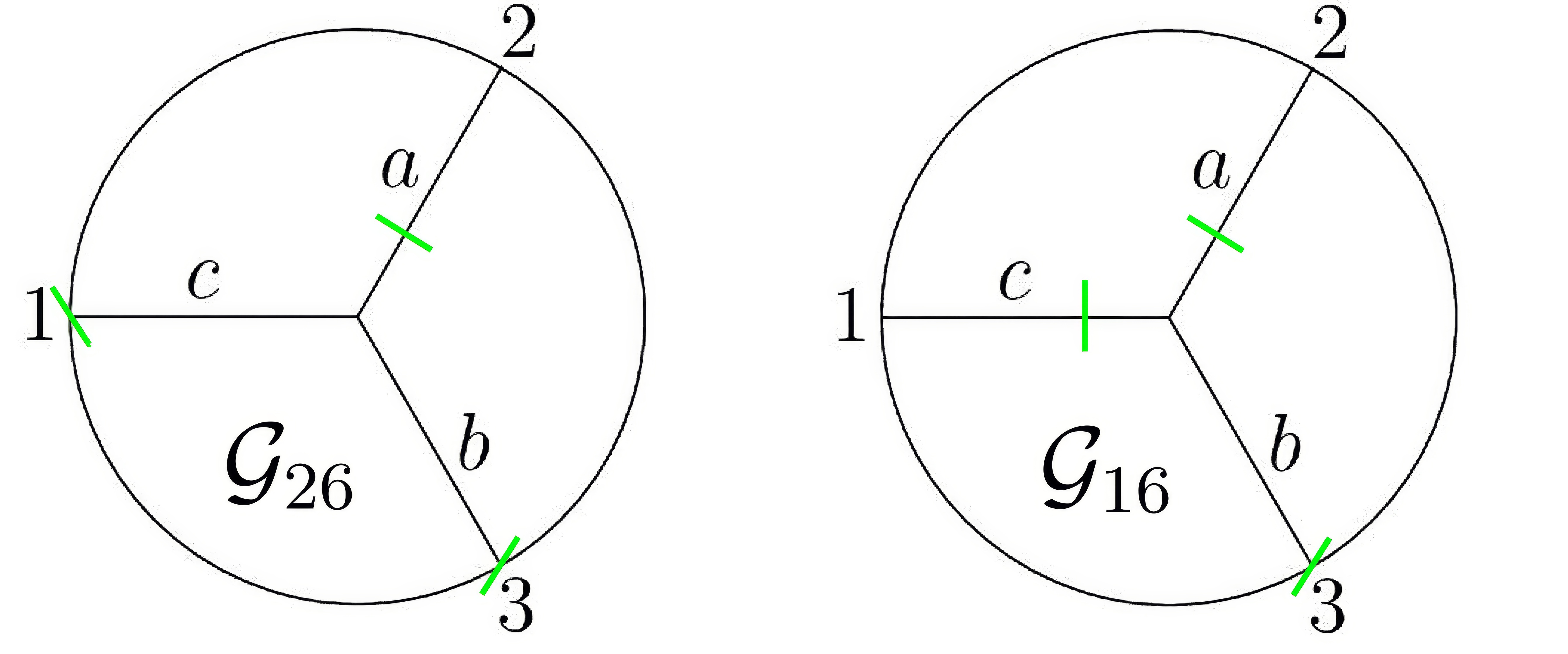}
     \caption{Minimal groves corresponding to faces labeled by $26$ and $16$} \label{fig:grove}
     \end{figure}
 \begin{example}
     Consider a network  $\mathcal{E}$ as it is shown in Fig. \ref{ex1}. Then, Formula \eqref{formtogrove} can be rewritten as follows:
     \begin{equation*}
         w(e_1)=\dfrac{\frac{1}{wt(\mathcal{G}_{26})}}{\frac{1}{wt(\mathcal{G}_{16})}}=\dfrac{\frac{1}{a}}{\frac{1}{ac}}=c.
     \end{equation*}
     
 \end{example}
 Proposition \ref{mainprop1} naturally leads to the following question:
\begin{problem}
    In \cite{MSp} minimal dimers $\Pi_F$ of a minimal  Lam model $N(\Gamma, \omega)$  were defined using strands of an oriented median graph $\Gamma_M$ of $N(\Gamma, \omega).$ Is there a similar approach for  constructing  minimal groves $\mathcal{G}_F$ of a minimal  electrical network $\mathcal{E}(G, w)$ with strands of  median graph $G_M$ of $\mathcal{E}(G, w)$?
\end{problem}

\subsection{Conductivity recovering via effective resistances}
In this section, we present a solution to an alternative formulation of Problem \ref{bl-box} that, beyond the motivations outlined in Section \ref{bl-box-sec}, also finds application in the theory of phylogenetic networks, as demonstrated in \cite{F1}, \cite{F2} and \cite{GK}. 
\begin{definition}
Let $\mathcal E(G, w) \in E_n$ be a connected  electrical network, and let the boundary voltages  $U = (U_1, \dots , U_n)$ be such that
\begin{equation*} \label{eq-resist}
    M_R(\mathcal E)U = -e_i + e_j,
\end{equation*}
 where $e_k, \ k \in \{1, \dots  , n\}$ is the  standard basis of $\mathbb{R}^n.$
 
Let us define the \textbf{effective resistance} $R_{ij}$ between  nodes $i$ and $j$ as follows $|U_i - U_j|:=R_{ij}.$   Effective resistances are well-defined and $R_{ij}=R_{ji}$. 

We will organize the effective resistances $R_{ij}$ in an  \textbf{effective resistance matrix}  $R_{\mathcal E}$ by setting $R_{ii}=0$ for all $i$.
\end{definition}
\begin{problem}  \label{bl-box-res}
    Consider an electrical network $\mathcal{E}(G, w)$ on a given graph $G.$ The conductivity function $w$ is required to be recovered by a known effective resistance matrix $R_{\mathcal E}$.
\end{problem}
The following theorem reveals the possibility of solving the described version of the black box problem:
\begin{theorem} \label{bl-box_th-res}	
The  conductivity function $w$ of a  minimal  electrical network $\mathcal{E}(G, w) \in E_n$ can be uniquely recovered by a known effective resistance matrix $R_{\mathcal E}$. 
\end{theorem}
We  provide a constructive proof of Theorem \ref{bl-box_th-res}, which   will be absolutely similar to the proof of  Theorem \ref{bl-box_th} due to the following result:
\begin{theorem} \textup{\cite{BGGK}}
Let $\mathcal E (G, w) \in E_n$ be a connected network. Then, using its effective resistance matrix $R_{\mathcal E}$, we define a point in $Gr(n-1,2n)$ associated with it as the row space of the matrix:
\begin{equation*} \label{eq:omega_n,r}
  \Omega_{R}(\mathcal E)=\left(\begin{matrix}
1 & m_{11} & 1 &  -m_{12} & 0 & m_{13} & 0 & \ldots  \\
0 & -m_{21} & 1 & m_{22} & 1 & -m_{23} & 0 & \ldots \\
0 & m_{31} & 0 & -m_{32} & 1 & m_{33} & 1 & \ldots \\
\vdots & \vdots & \vdots & \vdots & \vdots & \vdots &  \vdots & \ddots 
\end{matrix}\right),
\end{equation*}
where 
$$m_{ij}= -\frac{1}{2}(R_{i,j}+R_{i+1,j+1}-R_{i,j+1}-R_{i+1,j}).$$ 

Then, $\Omega_{R}(\mathcal E)$ defines the same point of $Gr_{\geq 0}(n-1,2n)$ as the point $\mathcal{L}(\mathcal{E}),$ see Theorem \ref{th: main_gr}. 
\end{theorem}

\begin{example}
    Consider a network as it is shown in Fig. \ref{pl-tree}. Then, its effective resistance matrix has the form   
    \[R_\mathcal{E}=
\begin{pmatrix}
0& 3 & 3& 2 \\
3&  0& 2& 3\\
3& 2&  0& 3 \\
2&3&3&0
\end{pmatrix},
\]
and the matrix $\Omega'_{R}(\mathcal{E})$ obtained from  $\Omega_{R}(\mathcal{E})$ by deleting its last row has the form
\[\Omega'_R(\mathcal{E})=
\begin{pmatrix}
1& 3 & 1& 1 & 0 & -1& 0 &1 \\
0& 1 & 1& 2 & 1 & 1& 0 &0\\
0& -1 & 0& 1 & 1 & 3& 1 &1 
\end{pmatrix}.
\]
\end{example}
\begin{figure}[H]
    \centering
    \includegraphics[width=0.4\textwidth]{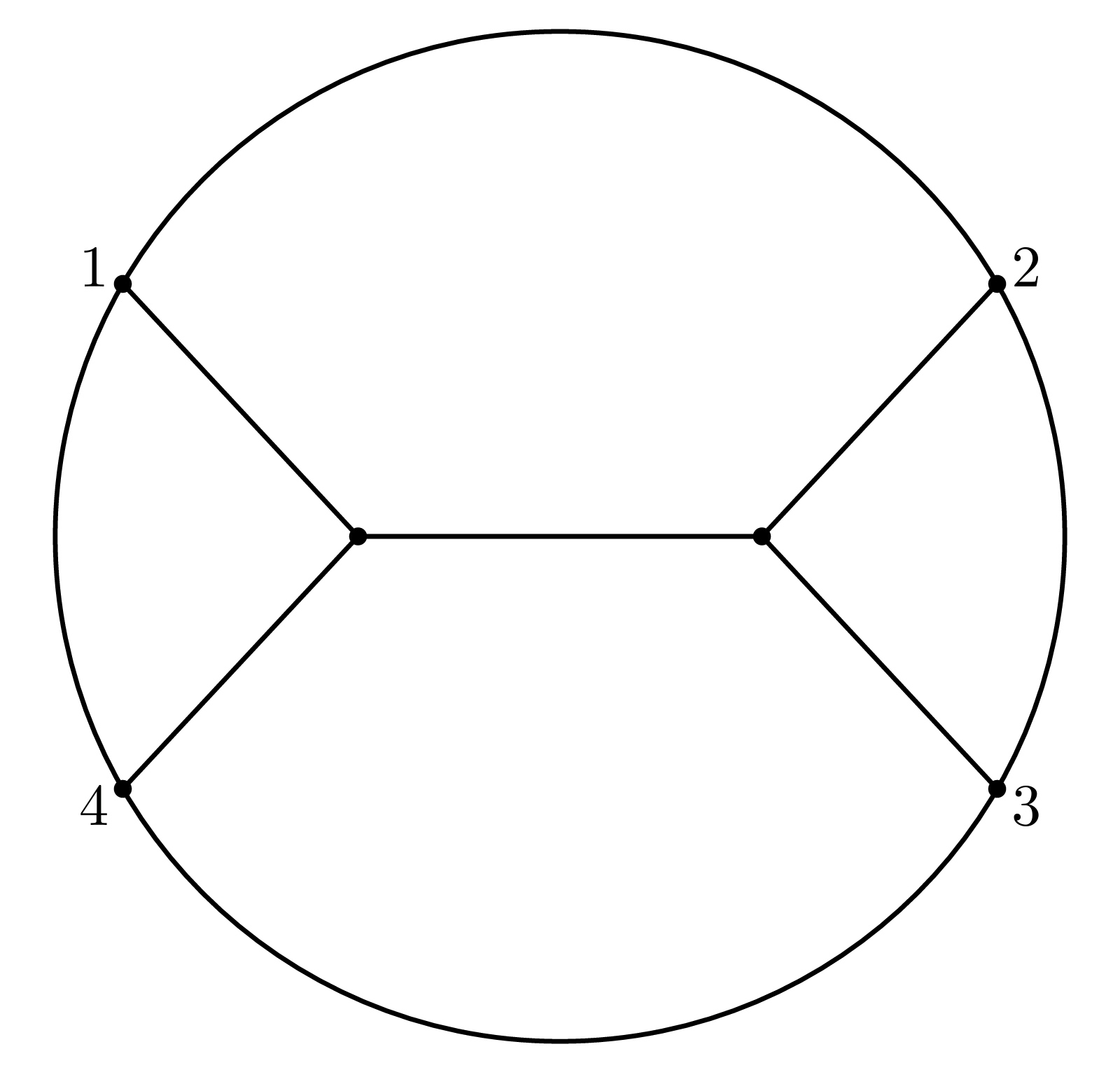}
    \caption{All conductivities are equal to $1$ }
    \label{pl-tree}
\end{figure}
Since the matrices $\Omega_R'(\mathcal{E})$ and $\Omega'(\mathcal{E})$ define the same point in $Gr_{\geq 0}(n-1, 2n)$, we immediately derive the following algorithm for solving Problem \ref{bl-box-res}:
\begin{algorithm} \label{mainalg-res}
\hskip 1pt
\begin{itemize}
    \item By a given effective resistance matrix $R_\mathcal{E}$ construct a matrix $\Omega_R'(\mathcal{E} )$;
    \item  Calculate $\tau\bigl(\Omega_R'(\mathcal{E})\bigr);$
    \item Do the last two steps of Algorithm \ref{mainalg}.
\end{itemize}
\end{algorithm}
\section{The black box problem as the   Berenstein-Fomin-Zelevinsky problem} \label{luzsec}
In this section, we discuss the connection between  Problem \ref{bl-box} and the classical theory of  totally non-negative matrices. To achieve this goal, we primarily describe an embedding of electrical networks into the positive part of the Lagrangian Grassmannian.
\begin{lemma} \textup{\cite{BGKT}}
Let us define the subspace $$V=\{v\in \mathbb{R}^{2n}| \sum_{i=1}^n(-1)^i v_{2i}=0\, \sum_{i=1}^n (-1)^iv_{2i-1}=0\}. $$  
Then,  the rows of  the matrix $\Omega(\mathcal{E})$  belong to the subspace $V$. Thus  the matrix $\Omega(\mathcal{E})$   defines a point of $Gr(n-1, V).$
\end{lemma}
By choosing a particular basis for the subspace $V$, we obtain an embedding of electrical networks into a Lagrangian Grassmannian $LG(n-1, V)$ with respect to a symplectic form that depends explicitly on the chosen basis of $V$. Moreover, specific choices (see  \cite{BGKT}, \cite{Tal}) of basis $V$ allow us to construct  embeddings of electrical networks into the positive part of the Lagrangian Grassmannians $LG_{\geq 0}(n-1, V):=LG(n-1, V) \cap Gr_{\geq 0}(n-1, V)$:
\begin{theorem} \textup{\cite{GoT},\cite{Tal}} \label{main-pos-matr}
When $n$ is odd, there exists a special basis of the subspace $V$ such that expanding the rows of the matrix $\Omega(\mathcal{E})$ with respect to this basis gives us a matrix $\overline{\Omega}(\mathcal{E})$ that defines a point of  $ LG_{\geq 0}(n-1, V)$ with  respect to the non-degenerate symplectic form:
\begin{equation*}
		\Lambda'_{2n-2} = \left(
		\begin{array}{cc}
			\Lambda_{n}& 0  \\
			0& -\Lambda_{n}\\
		\end{array}
		\right),
	\end{equation*}
    where $\Lambda_{n}$ is defined as in \eqref{form}.
\end{theorem}

   An electrical network is called \textbf{well-connected} if it is equivalent to the \textbf{standard} network, see \cite{K} and Fig. \ref{fig:standart}.

\begin{theorem} \textup{\cite{GoT},\cite{Tal}} \label{main-pos-matr1}
 Consider well-connected electrical networks with  an odd number of boundary nodes. Then, the matrix $\overline{\Omega}(\mathcal{E})$ and a matrix of the form described below represent the same point in $LG_{\geq 0}(n-1, V)$:
  \begin{equation*}
		\bigl(\overline{\mathrm{Id}},A(\mathcal{E})\bigr) = \left(
	\begin{array}{ccccccccc}
			0& \dots & 0 & 0  & 1 &a_{11} & a_{12} & a_{13} & \dots \\
			0& \dots & 0 & -1  & 0 & a_{21} & a_{22} & a_{23} & \dots \\
            0& \dots & 1 & 0  & 0 & a_{31} & a_{32} & a_{33} & \dots \\
	\vdots&\vdots&\vdots&\ddots&\vdots&\vdots& \vdots &  \ddots &\vdots
		\end{array}
		\right), \ \overline{Id} \in \mathrm{Mat}_{n-1 \times n-1}(\mathbb{R}).
	\end{equation*}
\end{theorem}
\begin{definition}
An invertible matrix $A \in \mathrm{Mat}_{n \times n}(\mathbb{R})$ is called    \textbf{totally non-negative (positive)} if all its minors are non-negative (positive).    
\end{definition}

\begin{theorem} \textup{\cite{GoT},\cite{Tal}} \label{main-pos-matr2}
Consider   a well-connected network $\mathcal{E}$ with an odd number of boundary nodes and a standard network $\mathcal{E'}$ that is equivalent to $\mathcal{E}$.  Then, the matrix $A(\mathcal{E})=(a_{ij}) \in \mathrm{Mat}_{n-1 \times n-1}(\mathbb{R})$ defined below is a totally positive symplectic matrix  with  respect to the skew-symmetric non-degenerate  form:
   \begin{equation} \label{form}
		\Lambda_n = \left(
	\begin{array}{ccccccccc}
			0& 1 & 0  & 0  & \dots   \\
			-1& 0 & -1  & 0 & \dots  \\
            0& 1 & 0  & 1 & \dots & \\
			\vdots&\vdots&\vdots&\ddots&\vdots
		\end{array}
		\right),  \ A(\mathcal{E})\Lambda_nA(\mathcal{E})^t=\Lambda_n.
	\end{equation}

 Moreover, the matrix $A(\mathcal{E})$ can be decomposed into the product of the elementary totally non-negative symplectic generators  $u_{i, i+1}(t)=E+t(E_{i+1, i}+E_{i-1, i})$:
   \begin{equation} \label{prod-dec}
    A(\mathcal{E})=\prod \limits_{1 \leq i <j \leq n}u_{j-i, j-i+1}(t_{ij}),  
 \end{equation}
 where  $E$ is the identity matrix, $E_{ij}$ is the matrix unity and $E_{n, n-1}=E_{0, 1}=0$; the product is taken over all pairs $(i,j)$ ordered in reverse lexicographical order; and $$t_{ij}=w_{ij}^{(-1)^{i+j}},$$ here $w_{ij}$ is a conductivity of  an edge of the network $\mathcal{E'}$ lying in intersection strands numbered $i$ and $j$, see Fig. \ref{numst}.
\end{theorem}
\begin{remark}
     The decomposition \eqref{prod-dec} is compatible with the bridge-spike decomposition, see \cite[Section $8.1$]{BGKT}. We suppose that based on this observation it is possible to extend Theorem \ref{main-pos-matr} and  Theorem \ref{main-pos-matr2}  to any network.
\end{remark}

\begin{figure}[H]
     \hspace*{-8mm}
     \includegraphics[scale=0.21]{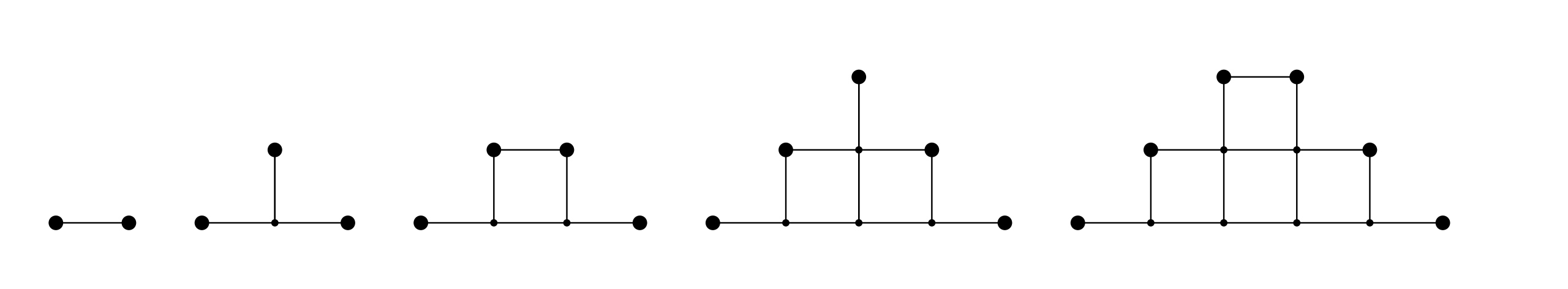}
     \caption{Standard networks, boundary nodes are bold } \label{fig:standart}
     \end{figure}

\begin{figure}[H]
    \centering
    \includegraphics[width=0.6\textwidth]{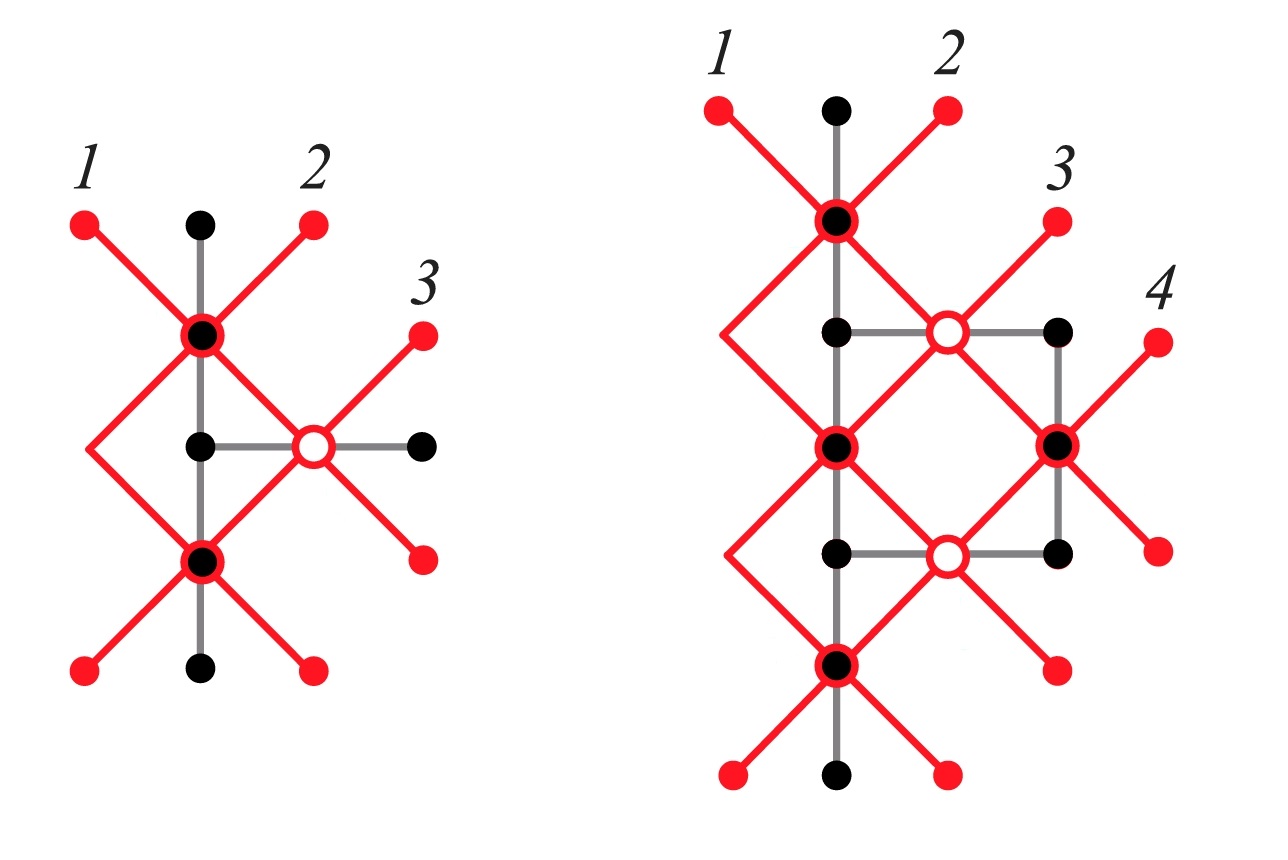}
    \caption{The edge enumeration}
    \label{numst}
\end{figure}

Theorem \ref{main-pos-matr2} provides an explicit link between electrical network theory and the theory of totally non-negative matrices. The cornerstone result of the last one is that: 
    \begin{theorem} \textup{\cite{Wh}} \label{aboutmatri}
        Each totally non-negative matrix can be decomposed into a product of elementary generators, i.e. Jacobi matrices $x_i(t)=E+tE_{i, i+1},$ $y_i(t)=E+tE_{i+1, i}$ and a diagonal matrix $d=\mathrm{diag}(d_1, \dots, d_n)$ with all non-negative parameters $t, d_i$.
    \end{theorem}
\begin{example}
The following holds:
   \begin{equation*}
		x_1(t) = \left(
	\begin{array}{ccc}
			1 & t & 0    \\
			0 & 1 & 0   \\
            0& 0 & 1 
		\end{array}
		\right), \
  y_2(t) = \left(
	\begin{array}{ccc}
			1 & 0 & 0    \\
			0 & 1 & 0   \\
            0& t & 1 
		\end{array}
		\right).
	\end{equation*} 
    The totally non-negative symplectic generators $u_{i,i+1}(t) \in \mathrm{Mat}_{n-1 \times n-1}(\mathbb{R})$ can be decomposed as follows:
     \begin{equation*}
	u_{1,2}(t)=y_1(t), \ u_{i,i+1}(t)=y_{i+1}(t)x_{i-1}(t), \  u_{n,n+1}(t)=x_{n-2}(t).
	\end{equation*}  
\end{example}
\begin{figure}[H]
    \centering
    \includegraphics[width=1.1\textwidth]{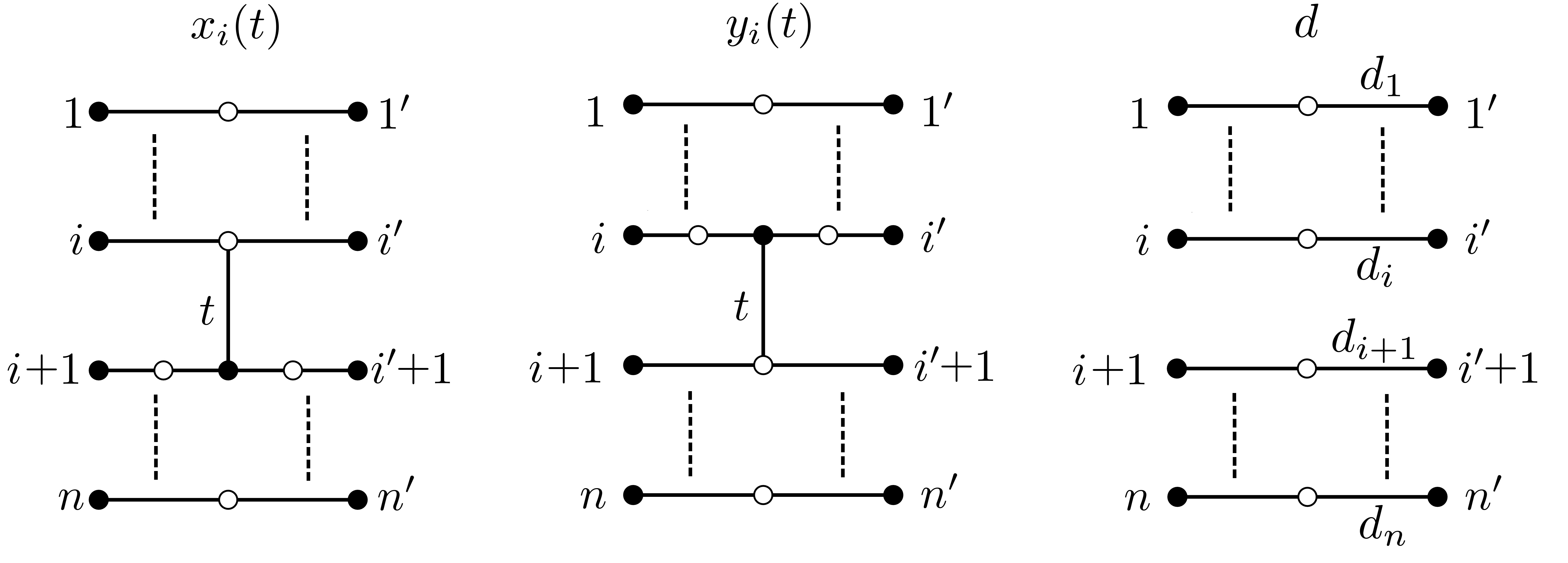}
    \caption{The elementary generators and their corresponded elementary Lam models}
    \label{elgenlam}
\end{figure}

Theorem \ref{aboutmatri} gives rise to the following natural problem, which, building upon the works  \cite{BFZ}, \cite{Lus}, we will call  the \textbf{Berenstein-Fomin-Zelevinsky problem}:
\begin{problem} \label{luz-prob}
    Consider a totally non-negative matrix $A$ and its given decomposition in a product of elementary generators. It is required  to recover the parameters of elementary generators.
\end{problem}

According to Theorem \ref{main-pos-matr2}, Problem \ref{bl-box} for well-connected networks with an odd number of boundary nodes can be considered  as a special case of Problem \ref{luz-prob}, which, in turn, is closely related to Problem \ref{bl-box-post}. This relationship is established by the following trick:
\begin{theorem} \textup{\cite{Pos}} \label{posmatr-lam}
    Let $A \in \mathrm{Mat}_{n \times n}(\mathbb{R})$  be a totally non-negative matrix with a known decomposition into a product of the elementary generators. Consider a point $X(A) \in Gr_{\geq 0}(n, 2n) $ which has the following form:
    \begin{equation*}
		X(A)=(\overline{\mathrm{Id}},A) = \left(
	\begin{array}{ccccccccc}
			0& \dots & 0 & 0  & 1 &a_{11} & a_{12} & a_{13} & \dots \\
			0& \dots & 0 & -1  & 0 & a_{21} & a_{22} & a_{23} & \dots \\
            0& \dots & 1 & 0  & 0 & a_{31} & a_{32} & a_{33} & \dots \\
	\vdots&\vdots&\vdots&\ddots&\vdots&\vdots& \vdots &  \ddots &\vdots
		\end{array}
		\right),
	\end{equation*}
    The point $X(A)$ is defined by a Lam model that is compatible with the decomposition of $A$ and can be built by concatenating \textbf{elementary Lam models}. Each elementary Lam model corresponds to an elementary generator, as it is shown in Fig. \ref{elgenlam} $($here, we use the non-standard notation for the boundary nodes numbering, which relates to an ordinary with the substitution $n \to 1, \dots, 1 \to n, 1' \to n+1, \dots, n' \to 2n)$.
\end{theorem}
\begin{example}
    Consider the standard electrical  network with $3$ boundary nodes. Using Theorem \ref{main-pos-matr}, we obtain that
    $$A(\mathcal{E})=u_{12}(t_{23})u_{23}(t_{13})u_{12}(t_{12})$$
    and the point $X\bigl(A(\mathcal{E})\bigr)$ is defined by the Lam model, as it is shown in Fig. \ref{spnetst}.
\end{example}
\begin{figure}[h!]
    \centering
    \includegraphics[width=0.6\textwidth]{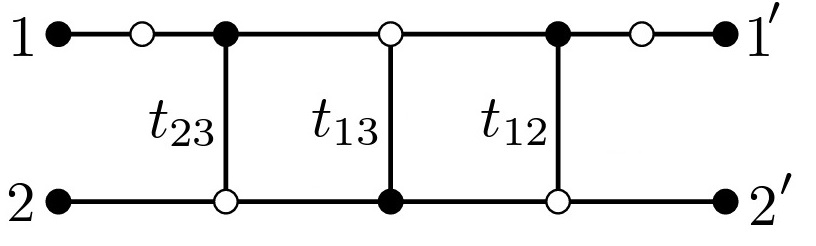}
    \caption{The Lam model associates with the standard electrical  network with $3$ boundary nodes}
    \label{spnetst}
\end{figure}
Note that the Lam models associated with points $X\bigl(A(\mathcal{E})\bigr)$ are generally not minimal. For instance, even the Lam model depicted in Fig. \ref{spnetst} is not minimal. Consequently, the recovery technique outlined in Section \ref{postteorsec} cannot be directly applied to solve the Berenstein-Fomin-Zelevinsky problem for totally non-negative symplectic matrices $A(\mathcal{E})$. However, we believe that articulating the connection described above between the black box problem and the Berenstein-Fomin-Zelevinsky problem might be useful. In particular, we expect the following to occur:     
\begin{conjecture}
The following algorithm for solving Problem \ref{bl-box} is correct:
\begin{itemize}
    \item Using a response matrix $M_R(\mathcal{E})$ of a  well-connected network $\mathcal{E}$ with an odd number of boundary node, construct the matrix $A(\mathcal{E})$;
    \item Using the decomposition \eqref{prod-dec} and Theorem \ref{posmatr-lam}, construct a Lam model. Using the Postnikov transformations (see \cite{L3}), into a minimal Lam model (we expect that is  multiple simple, canonical ways exist for this reduction);
    \item Using Theorem \ref{thabouttwit}, reconstruct the weights of the minimal Lam model. Reconstruct the weights of the original Lam model  from these;
    \item  Using the found weights and the physical sense of parameters of the decomposition \eqref{prod-dec}, construct a standard network $\mathcal{E'}$ equivalent to $\mathcal{E}$. Transform $\mathcal{E'}$ into $\mathcal{E}$ using star-triangle transformations.
\end{itemize}

\end{conjecture}

\end{document}